\documentclass[11pt]{article}
\usepackage{amssymb}
\usepackage{colortbl}
\usepackage{amsfonts,amsmath, longtable}
\usepackage{ulem}
\usepackage{comment}

\usepackage{hyperref}

\hypersetup{
    colorlinks = true,
    linktocpage = true,
    citecolor = {blue}
}

        \def\theequation{\thesection.\arabic{equation}}

\newcommand{\tr}{{\rm tr}}
\newcommand{\ti}[1]{\tilde{#1}}

\newcommand{\mL}{{\mathcal L}}
\newcommand{\mM}{{\mathcal M}}
\newcommand{\mF}{{\mathcal F}}

\newcommand{\mH}{{\mathcal H}}

\newcommand{\vf}{\varphi}

\newcommand{\ga}{\gamma}
\newcommand{\om}{\omega}
\newcommand{\vth}{\vartheta}

\newcommand{\Mat}{ {\rm Mat}(N,\mathbb C) }

\newcommand{\mC}{\mathbb C}
\newcommand{\mZ}{\mathbb Z}

\newtheorem{theorem}{Theorem}[section]
\newtheorem{proposition}{Proposition}[section]
\newtheorem{corollary}{Corollary}[section]
\newtheorem{predl}{Proposition}[section]
\newtheorem{lemma}{Lemma}[section]

\newenvironment{proof}{\par\noindent{\bf Proof.}}{\hfill$\scriptstyle\blacksquare$}
\def\beq{\begin{equation}}
\def\eq{\end{equation}}
\def\p{\partial}

\newcommand{\mat}[4]{\left(\begin{array}{cc}{#1}&{#2}\\ \ \\{#3}&{#4}
\end{array}\right)}

\newcommand{\mats}[4]{\left(\begin{array}{cc}{#1}&{#2}\\ {#3}&{#4}
\end{array}\right)}

\def\res{\mathop{\hbox{Res}}\limits}

\def\doubleunderline#1{\underline{\underline{#1}}}

\begin{document}

\begin{center}

\setcounter{page}{1}

\vspace{0mm}

{\LARGE{R-matrix valued Lax pair for elliptic Calogero-Inozemtsev system}}

\vspace{3mm}

{\LARGE{
and associative Yang-Baxter equations of ${\rm BC}_n$ type
}}

% \vspace{3mm}

%{\Large{\bf and its multispin generalizations}}

 \vspace{15mm}

 {\Large {M. Matushko}}
 %$\,^{\diamond}$$^\bullet$
\qquad
 {\Large {A. Mostovskii}}
 %$\,^{*}$
\qquad
 {\Large {A. Zotov}}
 %$\,^{\diamond}$$^{*}$

  \vspace{10mm}

%$\diamond$ --
 {\it Steklov Mathematical Institute of Russian
Academy of Sciences,\\ Gubkina str. 8, 119991, Moscow, Russia}

%$*$ --
%{\it Institute for Theoretical and Mathematical Physics,\\ Lomonosov Moscow State University, 119991,  Moscow, Russia}

%$\bullet$ --
%{\it Center for Advanced Studies, Skoltech, 143026, Moscow, Russia}

   \vspace{5mm}

 {\small\rm {E-mails: matushko@mi-ras.ru, mostovskii.am21@physics.msu.ru, zotov@mi-ras.ru}}

\end{center}

\vspace{0mm}

\begin{abstract}
We consider the elliptic Calogero-Inozemtsev system of ${\rm BC}_n$ type with five arbitrary constants and
propose $R$-matrix valued generalization for $2n\times 2n$ Takasaki's Lax pair.
For this purpose, we extend the Kirillov's ${\rm B}$-type associative Yang-Baxter equations to
similar relations depending on the spectral parameters and the Planck constants.
General construction uses the elliptic
Shibukawa-Ueno $R$-operator and the Komori-Hikami $K$-operators satisfying the reflection equation.
Then, using the Felder-Pasquier construction, the answer for the Lax pair is also
written in terms of the Baxter's 8-vertex $R$-matrix.
As a by-product of the constructed Lax pair we also propose a ${\rm BC}_n$ type generalization for
the elliptic XYZ long-range spin chain, and we present arguments pointing to its integrability.
\end{abstract}

\newpage
%\bigskip
{\small{ \tableofcontents }}

\bigskip

%\newpage

%%%%%%%%%%%%%%%%%%%%%%%%%%%%%%%%%%%%%%%%%%%%%%%%%%%%%%%%%%%%%%%%%%%%%%%%%%%%%%%%%%%%%%%%%%%%%%%%%
%%%%%%%%%%%%%%%%%%%%%%%%%%%%%%%%%%%%%%%%%%%%%%%%%%%%%%%%%%%%%%%%%%%%%%%%%%%%%%%%%%%%%%%%%%%%%%%%%
\section{Introduction}\label{sec1}
\setcounter{equation}{0}

\paragraph{Calogero-Moser model and $R$-matrices.} The
elliptic  Calogero-Moser model of ${\rm gl}_n$ type is
an integrable system of classical mechanics, describing pairwise interaction
of $n$ particles. It is defined by the Hamiltonian function \cite{OP}:
  \beq\label{q001}
  \begin{array}{c}
  \displaystyle{
H=\sum\limits_{i=1}^n\frac{p_i^2}{2}-{g}^2\sum\limits_{i>j}^n\wp(q_i-q_j)\,,
 }
 \end{array}
 \eq
 where $\wp(z)$ is the Weierstrass $\wp$-function, $g\in\mC$ is a coupling constant,
 $p_i\in\mC$ and $q_i\in\mC$ are the particles momenta and positions.
The Poisson brackets are canonical:
  \beq\label{q002}
  \begin{array}{c}
    \displaystyle{
\{p_i,q_j\}=\delta_{ij}\,,\quad \{p_i,p_j\}=\{q_i,q_j\}=0\,.
 }
 \end{array}
 \eq
The equations of motion
  \beq\label{q003}
  \begin{array}{c}
  \displaystyle{
 {\dot q}_i=p_i\,,\quad  {\ddot q}_i={g}^2\sum\limits_{k: k\neq i}^n\wp'(q_{i}-q_k)\,
 }
 \end{array}
 \eq
are represented in the Lax form
  \beq\label{q006}
  \begin{array}{c}
  \displaystyle{
{\dot L}(z)\equiv\{H,L(z)\}=[L(z),M(z)]
 }
 \end{array}
 \eq
with spectral parameter $z\in\mC$ (a coordinate on elliptic curve with moduli $\tau$, ${\rm Im}(\tau)>0$).
Explicit form for the Lax pair $L(z),M(z)$ was found in \cite{Krich1}
by I. Krichever.  The Lax matrix is of the size $n\times n$:
  \beq\label{q004}
  \begin{array}{c}
  \displaystyle{
%L^{\hbox{\tiny{CM}}}_{ij}(z)\equiv
L(z)=\sum\limits_{i,j=1}^n E_{ij}\,L_{ij}(z)\,,\quad
L_{ij}(z)=\delta_{ij}p_i+{g}(1-\delta_{ij})\phi(z,q_{ij})\,,\quad
q_{ij}=q_i-q_j\,,
 }
 \end{array}
 \eq
where $\phi(z,u)$ is the elliptic Kronecker function \cite{Weil}
 \beq\label{a001}
  \begin{array}{l}
  \displaystyle{
 \phi(z,u)=\frac{\vth'(0)\vth(z+u)}{\vth(z)\vth(u)}
 }
 \end{array}
 \eq
 defined through $\vth(z)$ -- the first Jacobi theta-function (\ref{a02}).
The expression for the $M$-matrix is given in the next Section. The validity of the Lax equation
is based on a set of elliptic function identities for the function $\phi(z,u)$.
The main one is the following summation formula (or the genus one Fay identity, see (\ref{a07})):
  \beq\label{q014}
  \begin{array}{c}
  \displaystyle{
\phi(z,q_{ik})\phi(w,q_{kj})=\phi(w,q_{ij})\phi(z-w,q_{ik})+\phi(w-z,q_{jk})\phi(z,q_{ij})\,.
 }
 \end{array}
 \eq
This relation is a particular case of a more general one known as the
associative Yang-Baxter equation (AYBE) introduced by S. Fomin and An. Kirillov in \cite{FK}. Let ${\mathcal A}$ be an associative algebra and $R^z(q_1,q_2)$ be a meromorphic function depending on parameters  $z,q_1,q_2\in \mC$ with values in  $\mathcal A \otimes \mathcal A$.  The AYBE with parameters proposed in \cite{Pol} is the following equation in $\mathcal A\otimes \mathcal A\otimes \mathcal A$:
  \beq\label{q2}
  \begin{array}{c}
  \displaystyle{
 R^z_{12}(q_1,q_2)
 R^{w}_{23}(q_2,q_3)=R^{w}_{13}(q_1,q_3)R_{12}^{z-w}(q_1,q_2)+R^{w-z}_{23}(q_2,q_3)R^z_{13}(q_1,q_3)\,,
 }
 \end{array}
 \eq
 where $R^z_{ij}(q_i,q_j)$ means the operator acting nontrivial as $R^z(q_i,q_j)$  on the $i$-th and the $j$-th tensor components. In case when
$R^z(q_1,q_2)$ depends on the
difference of spectral parameters
$R^z(q_1,q_2)=R^z(q_1-q_2)$ and $\mathcal A=\mC$ the AYBE (\ref{q2}) reproduces (\ref{q014}), where $R^z(q_1,q_2)$ acts by multiplication by the function $\phi(z,q_1-q_2)$.
In the matrix case $\mathcal A=\Mat$, the equation (\ref{q2}) is fulfilled by the elliptic Baxter-Belavin $R$-matrix \cite{Pol}.

It can be shown (see e.g. \cite{LOZ14})
that any solution of (\ref{q2}) satisfying certain
additional properties (see below the unitarity (\ref{q01}) the and skew-symmetry (\ref{q011}))
satisfies also the quantum Yang-Baxter equation
  \beq\label{q008}
  \begin{array}{l}
  \displaystyle{
 R^{\hbar}_{12}(q_1,q_2)  R^{\hbar}_{13}(q_1,q_3) R^{\hbar}_{23}(q_2,q_3) =
 R^{\hbar}_{23}(q_2,q_3)  R^{\hbar}_{13}(q_1,q_3)  R^{\hbar}_{12}(q_1,q_2)\,.
 }
 \end{array}
 \eq
 For this reason, we refer to any solution of the AYBE (\ref{q2})
 as $R$-matrix\footnote{In fact, the Yang-Baxter equations (\ref{q2}) and (\ref{q008})
 have different (and intersecting) sets of solutions. We consider a special class of solutions
 of AYBE (\ref{q2}), which also satisfy the quantum Yang-Baxter equation (\ref{q008}).}.
 Many possible applications of the AYBE to different algebraic and geometric constructions,
 including those related to integrable systems, can be found in \cite{FK,Pol,ORS,Kir}
  and \cite{LOZ14,LOZ15,LOZ16,SeZ18,GSZ,Z18,MZ}.

 The similarity between (\ref{q014}) and (\ref{q2}) was used in \cite{LOZ14} to propose the so-called
 $R$-matrix valued Lax pair, which is a Lax matrix with entries
  \beq\label{q009}
  \begin{array}{c}
  \displaystyle{
\mL_{ij}(z)={\rm Id}\,\delta_{ij}p_i+{g}(1-\delta_{ij})R^z_{ij}(q_i-q_j)\in\Mat^{\otimes n}\,.
 }
 \end{array}
 \eq
Different applications and properties of this Lax representation can be found in
\cite{LOZ14,LOZ15,LOZ16,SeZ18,GSZ,Z18}. In particular, in \cite{SeZ18} the anisotropic
(XYZ) version of the Haldane-Shastry-Inozemtsev long-range spin chain was
proposed using $R$-matrix valued Lax pair.

\paragraph{${\rm BC}_n$ type Calogero-Inozemtsev system and reflection equation.}
The elliptic Calogero-Moser model (\ref{q001}) is related to the root system ${\rm A}_{n-1}$
and can be extended to other root systems \cite{OP,DP}. The model with five arbitrary constants
was suggested by V. Inozemtsev \cite{Inoz89}. It is described by the Hamiltonian:
\begin{equation}
 \label{q0401}
   \displaystyle{
 H = \frac{1}{2} \sum_{k=1}^{n} p_k^2 - g^2
\sum_{i<j}^{n} \Big(\wp(q_i-q_j) + \wp(q_i+q_j)\Big) -
\frac{1}{2}\sum\limits_{a=0}^3\sum_{k=1}^{n} \nu_a^2\wp(q_k+\om_a)\,,
}
 \end{equation}
 where $\om_\ga$ are half-periods (\ref{w209}), and the five arbitrary constants are
 $g,\nu_0,\nu_1,\nu_2,\nu_3\in\mC$.
 Originally, the Lax representation was of size
$3n\times 3n$. We use $2n\times 2n$ Lax representation proposed by K. Takasaki in \cite{Ta}.
More precisely, we use the Lax pair in the form obtained by O.Chalykh in a conceptual way from Cherednik-Dunkl operators in \cite{Ch1}.

In our construction of the $R$-matrix valued Lax pair we also use  $K$-matrices $K^\hbar(q)\in \Mat$ \cite{IK} , which
are solutions to the reflection equation \cite{Skl-refl}:
  \beq\label{q010}
  \begin{array}{c}
  \displaystyle{
 R^-_{12}(q_1,q_2)K^\hbar_1(q_1) R^+_{12}(q_1,q_2)K^\hbar_2(q_2)=
 K^\hbar_2(q_2)R^+_{12}(q_1,q_2)K^\hbar_1(q_1)R^-_{12}(q_1,q_2)\,,
 }
 \end{array}
 \eq
where $R^-_{12}(q_1,q_2)=R^\hbar_{12}(q_1-q_2)$, $R^+_{12}(q_1,q_2)=R^\hbar_{12}(q_1+q_2)$.
The $K$-matrices play the role of the boundary conditions in quantum integrable systems.

In \cite{Hikami,KH,KH2} Y. Komori and K. Hikami suggested a solution of (\ref{q010})
with the operator valued Shibukawa-Ueno $R$-matrix
in the form of
certain operators ($K$-operators).
This was applied to construct a commuting set of quantum Hamiltonians for the model (\ref{q0401})
and its relativistic generalization (through the Dunkl-Cherednik type operators).
We actively use the results of \cite{Hikami,KH,KH2} in our consideration.

\paragraph{Purpose of the paper.}
The associative Yang-Baxter equation was written in the form (\ref{q2}) in \cite{Pol}, and is
also known as the {\it associative Yang-Baxter equation with parameters}, while
the original equation \cite{FK} appeared in the form of a set of relations (details are given
in Section \ref{sec4}) including
the main one
\beq\label{q0012}
\begin{array}{c}
  \displaystyle{
r_{ij} r_{jk}=r_{ik}r_{ij}+r_{jk}r_{ik} \qquad \text{ for distinct } i,j,k\,.
}
\end{array}
\eq
Then (\ref{q2}) can be considered as a representation of (\ref{q0012}) with parameters. In \cite{Kir}
the set of relations was extended to other root systems. For the root system of type $B$
the set of relations includes the four term relation
\beq\label{q0013}
  \displaystyle{
r_{ij}y_j=y_ir_{ij}+\tilde{r}_{ij}y_i+y_j\tilde{r}_{ij}\,.
}
\eq
In Section \ref{sec4} we propose its analogue with parameters and call it
 {\it the ${ BC}_n$ associative Yang-Baxter equation with parameters}.
 In particular, the analogue of the four term relation (\ref{q0013}) is as follows:
\beq\label{q0014}
  \displaystyle{
 R_{ij}^{w+z}(q_i-q_j)\tilde{K}^w_j(q_j)=\tilde{K}^w_i(q_i)R_{ij}^{z-w}(q_i-q_j)+
 \tilde{R}_{ij}^{w-z}(q_i+q_j)\tilde{K}^z_i(q_i)+
 \tilde{K}_j^{-z}(q_j)\tilde{R}_{ij}^{w+z}(q_i+q_j)\,.
 }
 \eq
 To obtain this result we use the operator-valued Shibukawa-Ueno $R$-matrix and
 the construction similar to the Komori-Hikami $K$-matrix.

Next, we show that the obtained relations allow to construct
$R$-matrix valued Lax pair for the Calogero-Inozemtsev model
by generalizing the Takasaki's $2n\times 2n$ Lax representation with
spectral parameter.

G. Felder and V. Pasquier
 in \cite{FP} explained that the Shibukawa-Ueno $R$-operator can be transformed into
the elliptic Baxter-Belavin $R$-matrix by specifying the space, where this operator acts.
Similarly, the $K$-operators from \cite{Hikami,KH} were transformed into the matrix
form in \cite{KH2}. We use these results and prove that the obtained relations for
the ${\rm BC}_n$ associative Yang-Baxter equation with parameters are valid
for the Baxter's 8-vertex $R$-matrix and the elliptic $K$-matrix.

Finally, using the approach of \cite{SeZ18} we present
the quantum Hamiltonian for an anisotropic long-range spin chain of ${\rm BC}_n$ type.

A list of notation used in the paper is given before the Appendix.

%%%%%%%%%%%%%%%%%%%%%%%%%%%%%%%%%%%%%%%%%%%%%%%%%%%%%%%%%%%%%%%%%%%%%%%%%%%%%%%%%%%%%%%%%%%%%%%%%
%%%%%%%%%%%%%%%%%%%%%%%%%%%%%%%%%%%%%%%%%%%%%%%%%%%%%%%%%%%%%%%%%%%%%%%%%%%%%%%%%%%%%%%%%%%%%%%%%
\section{Elliptic Calogero-Moser model and \texorpdfstring{$R$}{R}-matrices}\label{sec2}
\setcounter{equation}{0}

In this Section we recall main steps of construction of long-range spin chains from \cite{SeZ18}
based on the $R$-matrix valued Lax pairs \cite{LOZ14}. We begin with the famous classical
 Krichever's Lax pair \cite{Krich1}. Then we
describe the necessary properties of $R$-matrices which are then used as matrix elements in the
$R$-matrix valued Lax pairs. Finally, we explain that at equilibrium positions,
the classical Lax equation can be treated as a quantum Lax equation for some long-range spin chain.

\subsection{Krichever's Lax pair}

Consider the Lax pair for the
elliptic ${\rm gl}_n$ Calogero-Moser model (\ref{q001})-(\ref{q003}). The Lax pair of size $n\times n$,
and the Lax matrix has the form (\ref{q004}).
Similarly, for the matrix elements of $M$-matrix:
  \beq\label{q105}
  \begin{array}{c}
  \displaystyle{
M_{ij}(z)={g}\, d_i\delta_{ij} +{g}(1-\delta_{ij})f(z,q_{ij})\,,\quad
d_i=\sum\limits_{k: k\neq i}^n E_2(q_{ik})=-\sum\limits_{k: k\neq
i}^n f(0,q_{ik})\,,
 }
 \end{array}
 \eq
where $q_{ij}$ means $q_i-q_j$.
See the Appendix for elliptic functions definitions and properties.
The proof of the Lax equation (\ref{q006}) with the Lax pair (\ref{q004}), (\ref{q105})
is based on the identities (\ref{a09})-(\ref{a11}) written as
  \beq\label{q115}
  \begin{array}{c}
  \displaystyle{
\phi(z,q_{ik})f(z,q_{kj})-f(z,q_{ik})\phi(z,q_{kj})=
\phi(z,q_{ij})(f(0,q_{kj})-
f(0,q_{ik}))\,.
 }
 \end{array}
 \eq
and
  \beq\label{q1151}
  \begin{array}{c}
  \displaystyle{
\phi(z,q_{ij})f(z,q_{ji})-f(z,q_{ij})\phi(z,q_{ji})=\wp'(q_{ij})\,.
 }
 \end{array}
 \eq
These relations follow from the summation formula (genus one Fay identity) (\ref{a07}):
  \beq\label{q114}
  \begin{array}{c}
  \displaystyle{
\phi(z,q_{ik})\phi(w,q_{kj})=\phi(w,q_{ij})\phi(z-w,q_{ik})+\phi(w-z,q_{jk})\phi(z,q_{ij})\,.
 }
 \end{array}
 \eq

\subsection{\texorpdfstring{$R$}{R}-matrix identities}
Let us recall the main properties of the elliptic $R$-matrix which will be used for
the construction of the $R$-matrix valued Lax pair.
In what follows, we deal with the Shibukawa-Ueno (operator-valued)
and the Baxter-Belavin (matrix-valued) $R$-matrices.
These properties are valid for both versions.

The elliptic $R$-matrix satisfies the unitarity property
 \beq\label{q01}
  \begin{array}{l}
    \displaystyle{
R_{12}^\hbar(u)R_{21}^\hbar(-u)=(\wp(\hbar)-\wp(u)){\rm Id}\,,
 }
 \end{array}
 \eq
 where ${\rm Id}$ is either an identity operator or an identity matrix.
This is an $R$-matrix analogue of (\ref{a10}). Also, we have a skew-symmetry property
  \beq\label{q011}
    \displaystyle{
  R^\hbar_{12}(u)=-R^{-\hbar}_{21}(-u)\,,
  }
  \eq
which is an analogue of the simple relation $\phi(\hbar,u)=-\phi(-\hbar,-u)$.
The main relation is the associative
Yang-Baxter equation \cite{FK}
  \beq\label{q02}
  \begin{array}{c}
  \displaystyle{
 R^z_{ik}
 R^{w}_{kj}=R^{w}_{ij}R_{ik}^{z-w}+R^{w-z}_{kj}R^z_{ij}\,,\
 \ R^z_{ij}=R^z_{ij}(q_i\!-\!q_j)\,.
 }
 \end{array}
 \eq
This is a matrix generalization of the Fay identity (\ref{a07}) or (\ref{q114}), and
it is fulfilled by the elliptic Baxter-Belavin $R$-matrix \cite{Pol}.
Introduce also matrix analogue of the function $f(z,u)$ (\ref{a04}):
  \beq\label{q03}
  \begin{array}{c}
  \displaystyle{
F^{\,z}_{ij}(u)=\p_u R^{\,z}_{ij}(u)
 }
 \end{array}
 \eq
and
  \beq\label{q04}
  \begin{array}{c}
  \displaystyle{
F^{\,0}_{ij}(u)=F^{\,z}_{ij}(u)|_{z=0}=F^{\,0}_{ji}(-u)\,.
 }
 \end{array}
 \eq
Then similarly to (\ref{a09}) or (\ref{q115}) one gets
  \beq\label{q05}
  \begin{array}{c}
  \displaystyle{
R^z_{ik} F^z_{kj} - F^z_{ik} R^z_{kj} = F^0_{kj} R^z_{ij} - R^z_{ij}
F^0_{ik}\,.
 }
 \end{array}
 \eq
By differentiating (\ref{q01}) we also obtain
 \beq\label{q06}
 \begin{array}{c}
  \displaystyle{
 R^z_{ij} F^z_{ji} - F^z_{ij} R^z_{ji}=\wp'(q_{ij}){\rm Id}\,.
 }
 \end{array}
 \eq
%where ${\rm Id}$ is an identity matrix of an appropriate size.

%%%%%%%%%%%%%%%%%%%%%%%%%%%%%%%%%%%%%%%%%%%%%%%%%%%%%%%%%%%%%%%%%%%%%%%%%%%%%%%%%%%%%%%%%%%%%%%%%
\subsection{\texorpdfstring{$R$}{R}-matrix valued Lax pairs}

%\paragraph{${\rm gl}_n$ Calogero-Moser model.}
The Lax pair (\ref{q004}), (\ref{q105}) of the ${\rm gl}_n$ Calogero-Moser model
has the following generalization \cite{LOZ14} (see also \cite{GrZ18,SeZ18,Z18})
called the $R$-matrix valued Lax pair:
  \beq\label{q31}
  \begin{array}{c}
    \displaystyle{
{\mL}(z)=\sum\limits_{i,j=1}^n E_{ij}\otimes \mL_{ij}(z)\,,
 \quad\quad
\mL_{ij}(z)= {\rm Id}\, \delta_{ij}p_i
+g(1-\delta_{ij})R^{\,z}_{ij}(q_{ij})
 }
 \end{array}
 \eq
  \beq\label{q32}
  \begin{array}{c}
  \displaystyle{
\bar\mM_{ij}(z)= D_i\delta_{ij}
+g(1-\delta_{ij})F^{\,z}_{ij}(q_{ij})+\delta_{ij}\,\mF^{\,0}\,,\quad
D_i=-g\sum\limits_{k: k\neq i}^n F^{\,0}_{ik}(q_{ik})\,,
 }
 \end{array}
 \eq
where
  \beq\label{q33}
  \begin{array}{c}
  \displaystyle{
\mF^{\,0}
 =g\sum\limits_{k>m}^n F^{\,0}_{km}(q_{km})
 =\frac{g}{2}\sum\limits_{k\neq m}^n F^{\,0}_{km}(q_{km})\,.
 }
 \end{array}
 \eq
 If an operator $R^{z}(q)$ satisfies the properties from the previous section, namely (\ref{q01}), (\ref{q011}),  (\ref{q02}), the classical Lax equation  ${\dot\mL}(z)=[\mL(z),\bar\mM(z)]$ is equivalent to the classical equations of motion    (\ref{q003}) for the Calogero-Moser system.
The proof of this fact is based on the identities (\ref{q05})-(\ref{q06}) and the
following simple\footnote{Relation (\ref{q34}) does not use any specific $R$-matrix properties
but only $[R_{ab}^z,F^w_{cd}]=0$ for all distinct $a,b,c,d$.} relation:
  \beq\label{q34}
  \begin{array}{c}
  \displaystyle{
 [R^z_{ij},\mF^{\,0}]+\sum\limits_{l:l\neq i,j}R^z_{il} F^z_{lj} -
F^z_{il} R^z_{lj} = \sum\limits_{l:l\neq j} R^z_{ij} F^0_{lj} -
 \sum\limits_{l:l\neq i} F^0_{il} R^z_{ij}=D_iR_{ij}^z-R_{ij}^zD_j\,,\quad  \forall\ i\neq j.
 }
 \end{array}
 \eq
As a result, one gets the classical equations of motion (\ref{q003}). That is (\ref{q31})-(\ref{q32})
is just a multi-dimensional generalization of the Krichever's Lax pair (\ref{q009}), (\ref{q105}).

Originally, the $R$-matrix valued Lax pair was introduced using the elliptic Baxter-Belavin $R$-matrix for the ${\rm GL}_N$ Lie group, so that in (\ref{q31})-(\ref{q33})
 $R^z_{ij}(q_{ij})\in\Mat^{\otimes n}$,  ${\rm Id}$ is an identity matrix in $\Mat^{\otimes n}$ and $\mL(z)\in{\rm Mat}(n,\mC)\otimes\Mat^{\otimes n}$. Notice that all constructions are valid for the Shibukawa-Ueno $R$-operator, since it satisfies all the properties from the previous section.

Similarly, $R$-matrix valued Lax pairs can be deduced for other root systems. Such
results were described in \cite{GrZ18} by extending the Lax pairs \cite{OP,DP} for the
Calogero-Moser models related to the classical root systems of $B,C,D$ types.
	In accordance with the original construction from \cite{OP,DP} in these cases we have
not arbitrary coupling constants, they satisfy a set of constraints. In \cite{LOZ15}
the ${BC}_1$ case was considered with four arbitrary constants. In both papers
\cite{GrZ18} and \cite{LOZ15} an additional component of the quantum Hilbert space
was used in order to define matrix elements $\mL_{i,i+n}$ and $\mL_{n+i,i}$.
In the following, we use a different approach which is based on the usage of $K$-matrices.

\subsection{Long-range spin chains}
Here we assume that we deal with a matrix-valued $R$-operator:  $R^z_{ij}(q_{ij})\in\Mat^{\otimes n}$.
Having a $R$-matrix valued Lax pair one can deduce a quantum long-range spin chain
in the following way \cite{SeZ18}.
Denote by ${\mM}(z)$ the part of $\bar\mM(z)$ without
the term $\mF^0\in \Mat^{\otimes n}$, that is
  \beq\label{q3401}
  \begin{array}{c}
  \displaystyle{
\bar\mM(z)=\mM(z)+1_n\otimes\mF^0\,.
 }
 \end{array}
 \eq
 Next, consider the Lax equation
  \beq\label{q34011}
  \begin{array}{c}
  \displaystyle{
{\dot \mL}(z)=[\mL(z),\bar\mM(z)]
 }
 \end{array}
 \eq
at an equilibrium position, where
 $p_i={\dot q}_i=0$. The positions of particles are fixed at some points satisfying equations
 ${\dot p}_i=0$. For the elliptic Calogero-Moser model (\ref{q001}) this set of points is
  \beq\label{q3402}
  \begin{array}{c}
  \displaystyle{
q_j\rightarrow \zeta_j=\frac{j}{n}\,,\quad j=1,...,n\,.
 }
 \end{array}
 \eq
 Consider restrictions of the Lax matrices to the equilibrium position:
  \beq\label{q3403}
  \begin{array}{c}
  \displaystyle{
\mL'(z)=\mL(z)|_{q_j=\zeta_j,p_j=0}\,,\quad \mM'(z)=\mM(z)|_{q_j=\zeta_j}\,,\quad
{\mF^0}'=\mF^0|_{q_j=\zeta_j}\,.
 }
 \end{array}
 \eq
  Then from the Lax equation (taking into account that ${\dot \mL}(z)|_{q_j=\zeta_j}=0$) we get
  \beq\label{q3404}
  \begin{array}{c}
  \displaystyle{
 [{\mF^0}',\mL'(z)]=[\mL'(z),\mM'(z)]\,,
 }
 \end{array}
 \eq
which is a quantum Lax equation for a quantum model defined by the Hamiltonian ${\mF^0}'$.

 This construction can be regarded as
an explanation of the Polychronakos freezing trick \cite{P}, which relates the
quantum spin Calogero-Moser models with the long-range spin chains by restricting to the
classical equilibrium position of the corresponding spinless model. The above idea was used in \cite{SeZ18} to obtain
anisotropic version of the Inozemtsev long-range chain \cite{Inoz90}. Recently, an explanation of the freezing trick was suggested within the framework of hybrid integrable
systems in \cite{LRS}.
Let us also mention that isotropic trigonometric
${\rm BC}_n$ type long-range spin chains were described in \cite{EFGR} using the Dunkl operator
formalism.

Different generalizations and applications for this type of model can be found in \cite{Lam,MZ,IKT}
and references therein.

%%%%%%%%%%%%%%%%%%%%%%%%%%%%%%%%%%%%%%%%%%%%%%%%%%%%%%%%%%%%%%%%%%%%%%%%%%%%%%%%%%%%%%%%%%%%%%%%%
%%%%%%%%%%%%%%%%%%%%%%%%%%%%%%%%%%%%%%%%%%%%%%%%%%%%%%%%%%%%%%%%%%%%%%%%%%%%%%%%%%%%%%%%%%%%%%%%%
\section{Takasaki's \texorpdfstring{$2n\times 2n$}{2nx2n} Lax pair for Calogero-Inozemtsev system}\label{sec3}
\setcounter{equation}{0}
Here we recall the Lax pair for the Calogero-Inozemtsev model (\ref{q0401})
 with five arbitrary constants $g,\nu_0,\nu_1,\nu_2,\nu_3\in\mC$.
 Main purpose of this Section is to write down detailed verification
 of the classical Lax equation for the Takasaki's $2n\times 2n$ Lax pair.
 In the next Section this proof will be generalized to $R$-matrix valued case.

\subsection{The model and the Lax pair}
 Equations of motion are of the form:
\begin{equation}
 \label{q4011}
   \displaystyle{
{\dot p}_i={\ddot q}_i= g^2
\sum_{k:k\neq i}^{n} \Big(\wp'(q_i-q_k) + \wp'(q_i+q_k)\Big) +
\frac{1}{2}\sum\limits_{a=0}^3\nu_a^2\wp'(q_i+\om_a)\,.
}
 \end{equation}
In \cite{Inoz89} $3n\times 3n$ Lax representation was suggested. Some particular
cases (not arbitrary constants)
were previously known from \cite{OP}, see also \cite{DP}. Here we deal with $2n\times 2n$ Lax representation proposed by K. Takasaki in \cite{Ta}. Then O. Chalykh suggested a derivation of quantum and classical Lax pairs from Dunkl operators in \cite{Ch1}. His construction works for the Calogero-Inozemtsev system and for its relativistic generalization as well, so we use
the expression from \cite{Ch1}.
The description of the spectral curve can be found in \cite{Ch2},
and the case of $n=1$ was also studied in \cite{Z04}.

The Lax pair has a natural block-matrix structure:
  \beq\label{q402}
  \begin{array}{c}
  \displaystyle{
L(z)=\mat{L^{11}(z)}{L^{12}(z)}{L^{21}(z)}{L^{22}(z)}\,,\qquad
M(z)=\mat{M^{11}(z)}{M^{12}(z)}{M^{21}(z)}{M^{22}(z)}\,,
 }
 \end{array}
 \eq
where all entries $L^{ab}(z)$ and $M^{ab}(z)$, $a,b=1,2$ are $n\times n$ matrices (below $i,j=1,...,n$):
  \beq\label{q403}
  \begin{array}{c}
  \displaystyle{
L^{11}_{ij}(z)=\delta_{ij}p_i+g(1-\delta_{ij})\phi(z,q_{ij})\,,
 }
 \end{array}
 \eq
  \beq\label{q404}
  \begin{array}{c}
  \displaystyle{
L^{12}_{ij}(z)=\delta_{ij}v(z,q_i)+g(1-\delta_{ij})\phi(z,q^+_{ij})\,,
 }
 \end{array}
 \eq
  \beq\label{q405}
  \begin{array}{c}
  \displaystyle{
L^{21}_{ij}(z)=-\delta_{ij}v(-z,q_i)-g(1-\delta_{ij})\phi(-z,q^+_{ij})\,,
 }
 \end{array}
 \eq
  \beq\label{q406}
  \begin{array}{c}
  \displaystyle{
L^{22}_{ij}(z)=-\delta_{ij}p_i-g(1-\delta_{ij})\phi(-z,q_{ij})
 }
 \end{array}
 \eq
and
  \beq\label{q407}
  \begin{array}{c}
  \displaystyle{
M^{11}_{ij}(z)=\delta_{ij}a_i+g(1-\delta_{ij})f(z,q_{ij})\,,
 }
 \end{array}
 \eq
  \beq\label{q408}
  \begin{array}{c}
  \displaystyle{
M^{12}_{ij}(z)=\frac12\,\delta_{ij}v'(z,q_i)+g(1-\delta_{ij})f(z,q^+_{ij})\,,
 }
 \end{array}
 \eq
  \beq\label{q409}
  \begin{array}{c}
  \displaystyle{
M^{21}_{ij}(z)=\frac12\,\delta_{ij}v'(-z,q_i)+g(1-\delta_{ij})f(-z,q^+_{ij})\,,
 }
 \end{array}
 \eq
  \beq\label{q410}
  \begin{array}{c}
  \displaystyle{
M^{22}_{ij}(z)=\delta_{ij}a_i+g(1-\delta_{ij})f(-z,q_{ij})\,,
 }
 \end{array}
 \eq
where
the functions $\phi$ and $f$ are given by (\ref{q01}) and (\ref{a04}) respectively and
 the following notation are used:
  \beq\label{q411}
  \begin{array}{c}
  \displaystyle{
q_{ij}=q_i-q_j\,,\quad q^+_{ij}=q_i+q_j\,.
 }
 \end{array}
 \eq
The block-matrix (\ref{q402}) has the following (anti)symmetry properties:
\beq\label{Lanti}
L^{11}(z)=-L^{22}(-z), \qquad L^{12}(z)=-L^{21}(-z)\,,
\eq
\beq\label{Msym}
M^{11}(z)=M^{22}(-z), \qquad M^{12}(z)=M^{21}(-z)\,.
\eq
 The diagonal elements of $M^{11}(z)$ and $M^{22}(z)$ are given by expressions
  \beq\label{q412}
  \begin{array}{c}
  \displaystyle{
a_i=g\sum\limits_{k:k\neq i}^n\Big(\wp(q_{i}-q_k)+\wp(q_i+q_k)\Big)
+\frac12\sum\limits_{a=0}^3\nu_a\wp(q_i+\omega_a)
 }
 \end{array}
 \eq
with half-periods $\omega_a$ numerated as in (\ref{w209}). Using
the set of functions (\ref{w209})-(\ref{w210}) we also define
  \beq\label{q413}
  \begin{array}{c}
  \displaystyle{
v(z,u)\equiv v(z,u|\nu)=\sum\limits_{a=0}^3\nu_a\vf_a(2z,u+\om_a)=
\sum\limits_{a=0}^3\nu_a \exp(4\pi \imath z\p_\tau\om_a)\phi(2z,u+\om_a)
 }
 \end{array}
 \eq
and its partial derivative
  \beq\label{q414}
  \begin{array}{c}
  \displaystyle{
v'(z,u)=\p_uv(z,u)=\sum\limits_{a=0}^3\nu_a \exp(4\pi \imath z\p_\tau\om_a)f(2z,u+\om_a)\,.
 }
 \end{array}
 \eq
Some useful formulae for the functions $v(z,u)$ and $v'(z,u)$ can be found in
 \cite{KH,KH2,Hikami} and \cite{Z04}. In particular, we have
  \beq\label{q415}
  \begin{array}{c}
  \displaystyle{
v(z,u)v(z,-u)=\sum\limits_{a=0}^3\Big(\bar\nu_a^2\wp(z+\om_a)-\nu_a^2\wp(u+\om_a)\Big)\,,
 }
 \end{array}
 \eq
 where similarly to (\ref{w215})
  \beq\label{q416}
 \begin{array}{c}
  \left(\begin{array}{c}
 \bar\nu_0
 \\
 \bar\nu_1
 \\
 \bar\nu_2
 \\
 \bar\nu_3
 \end{array}\right)
 =
   \displaystyle{\frac12}
 \left(\begin{array}{cccc}
 1 & 1 & 1 & 1
 \\
  1 & 1 & -1 & -1
 \\
  1 & -1 & 1 & -1
 \\
  1 & -1 & -1 & 1
 \end{array}\right)
  \left(\begin{array}{c}
 \nu_0
 \\
 \nu_1
 \\
 \nu_2
 \\
 \nu_3
 \end{array}\right)\,.
 \end{array}
 \eq
From (\ref{w215}) we conclude that
  \beq\label{q419}
  \begin{array}{c}
  \displaystyle{
v(z,u|\nu)=v(u,z|\bar\nu)=\sum\limits_{a=0}^3\bar\nu_a\vf_a(2u,z+\om_a)\,.
 }
 \end{array}
 \eq
By differentiating (\ref{q415}) with respect to $u$ we get:
  \beq\label{q420}
  \begin{array}{c}
  \displaystyle{
v(z,u)v'(z,-u)-v(z,-u)v'(z,u)=\sum\limits_{a=0}^3\nu_a^2\wp'(u+\om_a)\,.
 }
 \end{array}
 \eq
The calculation of the Hamiltonian (\ref{q0401}) is standard. Using (\ref{a10}) and (\ref{q415})
one gets
  \beq\label{q421}
  \begin{array}{c}
  \displaystyle{
\frac14\,\tr(L^2(z))=H+n(n-1)g^2\wp(z)+\frac{n}{2}\sum\limits_{a=0}^3\bar\nu_a^2\wp(z+\om_a)\,.
 }
 \end{array}
 \eq
In order to prove the Lax equation for the Lax pair (\ref{q402}) it
is helpful to use the following identities:
\beq\label{hident1}
  \begin{array}{c}
  \displaystyle{
v(x,u)\phi(x+y,w-u)+v(x,w)\phi(x-y,u-w)+v(y,-u)\phi(x+y,u+w)=
 }
 \\ \ \\
   \displaystyle{
=v(y,w)\phi(x-y,u+w)
}
\end{array}
\eq
and
\beq\label{hident3}
\begin{array}{c}
\phi(z,u-w)(v'(0,w)-v'(0,u))=
\\ \\
2v(-z,w)f(z,u+w)+2v(z,u)f(-z,u+w)+v'(-z,w)\phi(z,u+w)+v'(z,u)\phi(-z,u+w)\,,
\end{array}
\eq
where
  \beq\label{q441}
  \begin{array}{c}
  \displaystyle{
v'(0,u)=\sum\limits_{a=0}^3\nu_a f(0,u+\om_a)=-\sum\limits_{a=0}^3\nu_a E_2(u+\om_a)=
-\sum\limits_{a=0}^3\nu_a \wp(u+\om_a)+\frac{\vth'''(0)}{3\vth'(0)}\sum\limits_{a=0}^3\nu_a\,.
 }
 \end{array}
 \eq
The latter relation is derived from (\ref{a071}) and (\ref{a09}),
and the identities (\ref{hident1}) and (\ref{hident3}) are proved in the Appendix.

\subsection{Lax equation }\label{sec32}
Let us verify the Lax equation (\ref{q006}) in detail in order to generalize it in the Section \ref{sec5}.
We are going to use the identity (\ref{hident3}) in the form:
  \beq\label{q440}
  \begin{array}{c}
  \displaystyle{
 \phi(z,q_{ij}^+)v'(z,-q_j)-v'(z,q_i)\phi(z,-q_{ij}^+)+2v(z,q_i)f(z,-q_{ij}^+)-2f(z,q_{ij}^+)v(z,-q_j)=
 }
 \\ \ \\
  \displaystyle{
 =\phi(z,q_{ij})(v'(0,q_j)-v'(0,q_i))\,.
 }
 \end{array}
 \eq
It is natural to subdivide verification by considering four $n\times n$ matrix blocks separately.
\paragraph{The block 11.}
We need to show that
\beq\label{block11}
\dot{L}^{11}(z)=L^{11}(z)M^{11}(z)-M^{11}(z)L^{11}(z)
+L^{12}(z)M^{21}(z)-M^{12}(z)L^{21}(z)\,.
\eq
For the diagonal part we have:
\beq\label{q422}
  \begin{array}{c}
  \displaystyle{
\dot{L^{11}_{ii}}=L^{12}_{ii}M^{21}_{ii}-
M^{12}_{ii}L^{21}_{ii}+\sum\limits_{k:k\neq i}^n
\Big(L^{12}_{ik}M^{21}_{ki}-M^{12}_{ik}L^{21}_{ki}+
L^{11}_{ik}M^{11}_{ki}-M^{11}_{ik}L^{11}_{ki}\Big)\,,
 }
 \end{array}
 \eq
or, explicitly:
  \beq\label{q423}
  \begin{array}{c}
  \displaystyle{
\dot{p_{i}}=\frac{1}{2}v(z,q_{i})v'(z,-q_{i})-\frac{1}{2}v(z,-q_{i})v'(z,q_{i})+
 }
\\ \ \\
   \displaystyle{
+g^{2}\sum\limits_{k:k\neq i}^n\Big(\phi(z,q^{+}_{ik})f(z,-q^{+}_{ik})-\phi(z,-q^{+}_{ik})f(z,q^{+}_{ik})+
 }
\\ \ \\
   \displaystyle{
+\phi(z,q_{ik})f(z,q_{ki})-\phi(z,q_{ki})f(z,q_{ik})\Big)=
 }
\\ \ \\
   \displaystyle{
=\frac{1}{2}\sum\limits_{a=0}^3\nu_a^2\wp'(q_{i}+\om_a)+
g^{2}\sum\limits_{k:k\neq i}^n\Big(\wp'(q^{+}_{ik})+\wp'(q_{ik})\Big)\,,
 }
 \end{array}
 \eq
where we used (\ref{q420}) and (\ref{q1151}).

For the non-diagonal part we have:
\beq\label{q424}
  \begin{array}{c}
  \displaystyle{
\dot{L^{11}_{ij}}=L^{12}_{ij}M^{21}_{jj}-M^{12}_{ij}L^{21}_{jj}+
L^{11}_{ij}M^{11}_{jj}-M^{11}_{ij}L^{11}_{jj}+L^{12}_{ii}M^{21}_{ij}-M^{12}_{ii}L^{21}_{ij}+
L^{11}_{ii}M^{11}_{ij}-M^{11}_{ii}L^{11}_{ij}+
 }
 \\ \ \\
 \displaystyle{
 +\sum\limits_{k:k\neq i,j}^n\Big(L^{12}_{ik}M^{21}_{kj}-M^{12}_{ik}L^{21}_{kj}+
 L^{11}_{ik}M^{11}_{kj}-M^{11}_{ik}L^{11}_{kj}\Big)\,.
 }
 \end{array}
 \eq
The l.h.s. equals $(\dot{q_{i}}-\dot{q_{j}})f(z,q_{ij})$. The r.h.s is
\beq\label{q425}
  \begin{array}{c}
  \displaystyle{
%(\dot{q_{i}}-\dot{q_{j}})f(z,q_{ij})=
f(z,q_{ij})(p_{i}-p_{j})+
\doubleunderline{\frac{1}{2}\,\phi(z,q^{+}_{ij})v'(z,-q_{j})+f(z,q^{+}_{ij})v(-z,q_{j})+}
}
 \\ \ \\
 \displaystyle{
\doubleunderline{+\frac{1}{2}\,\phi(-z,q^{+}_{ij})v'(z,q_{i})+f(z,-q^{+}_{ij})v(z,q_{i})+}
}
 \\ \ \\
 \displaystyle{
 +\underline{g\sum\limits_{k:k\neq i,j}^n\phi(z,q_{ij})
 \Big(\wp(q_{jk})+\wp(q^{+}_{jk})-\wp(q_{ik})-\wp(q^{+}_{ik})\Big)}+
 }
 \\ \ \\
 \displaystyle{
 \doubleunderline{+\frac12\sum\limits_{a=0}^3\nu_a\phi(z,q_{ij})
 \Big(\wp(q_j+\omega_a)-\wp(q_i+\omega_a)\Big)+}
 }
 \\ \ \\
 \displaystyle{
 +\underline{g\sum\limits_{k:k\neq i,j}^n\Big(\phi(z,q^{+}_{ik})f(z,-q^{+}_{kj})
 -\phi(z,-q^{+}_{kj})f(z,q^{+}_{ki})
 +\phi(z,q_{ik})f(z,q_{kj})-\phi(z,q_{kj})f(z,q_{ik})\Big)} =
 }
 \\ \ \\
 \displaystyle{
 =f(z,q_{ij})(p_{i}-p_{j})\,,
 }
 \end{array}
 \eq
where the underlined and the double underlined terms are cancelled out
respectively due to (\ref{a09}) and (\ref{hident3}), (\ref{q441}).

\paragraph{The block 12.} Here
\beq\label{block12}
\dot{L}^{12}(z)=L^{11}(z)M^{12}(z)+L^{12}(z)M^{22}(z)-M^{11}(z)L^{12}(z)-M^{12}(z)L^{22}(z)\,.
\eq
For the diagonal part we have:
\beq\label{q426}
  \begin{array}{c}
  \displaystyle{
\dot{L^{12}_{ii}}=L^{11}_{ii}M^{12}_{ii}-
M^{12}_{ii}L^{22}_{ii}+\sum\limits_{k:k\neq i}^n
\Big(L^{11}_{ik}M^{12}_{ki}-M^{11}_{ik}L^{12}_{ki}+
L^{12}_{ik}M^{22}_{ki}-M^{12}_{ik}L^{22}_{ki}\Big)\,,
 }
 \end{array}
 \eq
that is
\beq\label{q427}
  \begin{array}{c}
  \displaystyle{
\dot{q_{i}}v'(z,q_{i})=p_{i}v'(z,q_{i})+g^{2}\sum\limits_{k:k\neq i}^n\Big(\phi(z,q_{ik})f(z,q^{+}_{ki})-f(z,q_{ik})\phi(z,q^{+}_{ki})+
 }
\\ \ \\
\displaystyle{
\phi(z,q^{+}_{ik})f(z,q_{ik})-f(z,q^{+}_{ik})\phi(z,q_{ik})\Big)=p_{i}v'(z,q_{i})\,.
 }
 \end{array}
 \eq
For the non-diagonal part we have:
\beq\label{q428}
  \begin{array}{c}
  \displaystyle{
\dot{L^{12}_{ij}}=L^{11}_{ij}M^{12}_{jj}-M^{11}_{ij}L^{12}_{jj}+
L^{11}_{ii}M^{12}_{ij}-M^{11}_{ii}L^{12}_{ij}+L^{12}_{ii}M^{22}_{ij}-
M^{12}_{ii}L^{22}_{ij}+L^{12}_{ij}M^{22}_{jj}-M^{12}_{ij}L^{22}_{jj}+
 }
 \\ \ \\
 \displaystyle{
+\sum\limits_{k:k\neq i,j}^n\Big(L^{11}_{ik}M^{12}_{kj}-M^{11}_{ik}L^{12}_{kj}+
L^{12}_{ik}M^{22}_{kj}-M^{12}_{ik}L^{22}_{kj}\Big)\,.
 }
 \end{array}
 \eq
The explicit form can be obtained from (\ref{q425}) by replacing $q_j, p_j$ to $-q_j,-p_j$.

\paragraph{The block 21 and 22}
For other blocks we need to show that
\beq\label{block21}
\dot{L}^{21}(z)=L^{21}(z)M^{11}(z)+L^{22}(z)M^{21}(z)-M^{21}(z)L^{11}(z)-M^{22}(z)L^{21}(z),
\eq
\beq\label{block22}
\dot{L}^{22}(z)=L^{21}(z)M^{12}(z)+L^{22}(z)M^{22}(z)-M^{21}(z)L^{12}(z)-M^{22}(z)L^{22}(z),
\eq
which follow from (\ref{block12}) and  (\ref{block11}) respectively due
to the properties  (\ref{Lanti}) and  (\ref{Msym}).

%%%%%%%%%%%%%%%%%%%%%%%%%%%%%%%%%%%%%%%%%%%%%%%%%%%%%%%%%%%%%%%%%%%%%%%%%%%%%%%%%%%%%%%%%%%%%%%%%
%%%%%%%%%%%%%%%%%%%%%%%%%%%%%%%%%%%%%%%%%%%%%%%%%%%%%%%%%%%%%%%%%%%%%%%%%%%%%%%%%%%%%%%%%%%%%%%%%
\section{Shibukawa–Ueno elliptic \texorpdfstring{$R$}{R}-operator and \texorpdfstring{$K$}{K}-matrices}\label{sec4}
\setcounter{equation}{0}
In this Section we formulate the parametric version of the Kirillov's $B$-type associative Yang-Baxter equations.
We show that these relations are fulfilled with the
operator-valued Shibukawa-Ueno elliptic $R$-matrix and the Komori-Hikami $K$-matrix.
\subsection{Associative Yang-Baxter equation and \texorpdfstring{$K$}{K}-operators}
\setcounter{equation}{0}
In \cite{SU} Y. Shibukawa and K.Ueno introduced an elliptic operator $R_{ij}$:
\beq\label{SUope}
R^\hbar_{ij}(q)=\phi(\hbar,x_i-x_j)-\phi(q,x_i-x_j)\hat{s}_{ij}\,,
\eq
which
satisfies the quantum Yang-Baxter equation
\beq\label{qYB}
R^\hbar_{ij}(q_i-q_j)R^\hbar_{ik}(q_i-q_k)R^\hbar_{jk}(q_j-q_k)=
R^\hbar_{jk}(q_j-q_k)R^\hbar_{ik}(q_i-q_k)R^\hbar_{ij}(q_i-q_j).
\eq
The operator $R_{ij}$ acts on a space of meromorphic functions of variables $x_1,\dots x_n$ (instead of a tensor product of vector spaces).
Here, the symmetric group $S_N$ acts by exchanging the variables $x_i$ and $\hat{s}_{ij}$ denotes the corresponding transposition
\beq\label{b24}
\hat{s}_{ij}f(\dots,x_i,\dots,x_j,\dots)=f(\dots,x_j,\dots,x_i,\dots).
\eq
%
%\begin{lemma}
 %The operator (\ref{SUope}) satisfies the associative Yang-Baxter equation:
% \beq\label{AYBE1}
 %R_{12}^z(q_1-q_2) R_{23}^w(q_2-q_3)=R_{13}^w(q_1-q_3)R_{12}^{z-w}(q_1-q_2)+R_{23}^{w-z}(q_2-q_3)R_{13}^z(q_1-q_3).
 %\eq
%\end{lemma}
%
In \cite{KH2} Y. Komori and K. Hikami proposed the following boundary $K$-operator:
\beq\label{K}
K^\hbar_i(q) \equiv K^\hbar_i(q|\bar\nu)=v(q,x_i|\bar\nu)-v(\hbar,x_i|\bar\nu)\hat{t}_i\,,
\eq
where $\hat{t}_i$ is a reflection operator
\beq\label{tacts}
\hat{t}_i f(x_1\dots,x_i,\dots,x_n)=f(x_1\dots,-x_i,\dots,x_n)\,.
\eq
The operator relations are as follows:
\beq\label{operrel}
\hat{t}_i^2=1,\qquad \hat{s}_{ij}\hat{t}_j=\hat{t}_i  \hat{s}_{ij}\,.
\eq
 It was also shown in \cite{KH2} that the reflection equation holds true for these operators, that is
\beq\label{REKH}
R_{ij}^{\hbar}(q_i-q_j)K_i^\hbar(q_i)R_{ji}^{\hbar}(q_i+q_j)K_j^\hbar(q_j)=
K_j^\hbar(q_j)R_{ij}^{\hbar}(q_i+q_j)K_i^\hbar(q_i)R_{ji}^{\hbar}(q_i-q_j)\,.
\eq

Let us introduce the following operator:
\beq\label{Rtilde}
\tilde{R}_{ij}^\hbar(q)=\phi(\hbar,x_i+x_j)-\phi(q,x_i+x_j)\hat{s}_{ij}\hat{t}_i\hat{t}_j\,,
\eq
where the operators $\tilde{R}_{ij}^{\hbar}(q)$
and $R_{ij}^{\hbar}(q)$ are connected by conjugation by $\hat{t}_j$:
\beq\label{b25}
\tilde{R}_{ij}^{\hbar}(q)=\hat{t}_j R_{ij}^\hbar(q) \hat{t}_j\,.
\eq
Conjugating  (\ref{qYB}) by $\hat{t}_k$ ad using (\ref{operrel}) we have the following equation:
\beq\label{b25a}
R^\hbar_{ij}(q_i-q_j)\tilde{R}^\hbar_{ik}(q_i-q_k)\tilde{R}^\hbar_{jk}(q_j-q_k)=
\tilde{R}^\hbar_{jk}(q_j-q_k)\tilde{R}^\hbar_{ik}(q_i-q_k)R^\hbar_{ij}(q_i-q_j).
\eq
Also, we have:
\beq\label{b26}
\tilde{R}_{ij}^{\hbar}(q)=-\hat{t}_i R_{ij}^{-\hbar}(-q) \hat{t}_i\,.
\eq
Introduce another operator $\tilde{K}$:
\beq\label{Ktilde}
\tilde{K}^\hbar_i(q) \equiv \tilde{K}^\hbar_i(q|\bar\nu)=v(\hbar,x_i|\bar\nu)-v(q,x_i|\bar\nu)\hat{t}_i\,.
\eq
It is related to the operator (\ref{K}) as
\beq\label{b27}
\tilde{K}^\hbar_i(q)=-{K}^\hbar_i(q)\hat{t}_i=\hat{t}_i K^{-\hbar}_i(-q)\,.
\eq
Using (\ref{b25}), (\ref{b26}) and (\ref{b27}),
we can rewrite the reflection equation (\ref{REKH}) in terms of $\tilde{K}^\hbar_i(q)$:
\beq\label{RE}
R_{ij}^{\hbar}(q_i-q_j)\tilde{K}_i^\hbar(q_i)\tilde{R}_{ij}^{\hbar}(q_i+q_j)
\tilde{K}_j^\hbar(q_j)=\tilde{K}_j^\hbar(q_j)\tilde{R}_{ij}^{\hbar}(q_i+q_j)
\tilde{K}_i^\hbar(q_i)R_{ij}^{\hbar}(q_i-q_j)\,.
\eq
\begin{lemma}
 The following unitarity properties hold true:
   \beq\label{uni}
  R_{ij}^{\hbar}(q)  R_{ji}^{\hbar}(-q)=(\wp(\hbar)-\wp(q))\,,
  \eq
  \beq\label{uni2}
  \tilde{R}_{ij}^{\hbar}(q)  \tilde{R}_{ji}^{-\hbar}(q)=(\wp(q)-\wp(\hbar))\,,
  \eq
   \beq\label{uniK}
 K^\hbar_i(q) K^\hbar_i(-q)=\sum_{a=0}^3 \nu^2_a\left(\wp(\hbar+\omega_a)-\wp(q+\omega_a) \right)\,,
 \eq
 \beq\label{unitildeK}
 \tilde{K}^\hbar_i(q) \tilde{K}^{-\hbar}_i(q)=\sum_{a=0}^3 \nu^2_a\left(\wp(q+\omega_a)-\wp(\hbar+\omega_a) \right)\,.
 \eq
\end{lemma}
\begin{proof}
Let us prove (\ref{uni}):
\beq\label{b27a}
\begin{array}{ccc}
R_{ij}^{\hbar}(q)  R_{ji}^{\hbar}(-q)=\left( \phi(\hbar,x_i-x_j)-\phi(q,x_i-x_j)s_{ij}\right)\left( \phi(\hbar,x_j-x_i)-\phi(-q,x_j-x_i)s_{ij}\right)=
\\ \\
=\phi(\hbar,x_i-x_j)\phi(\hbar,x_j-x_i)-\phi(q,x_i-x_j) \phi(\hbar,x_i-x_j)s_{ij}
\\ \\
-\phi(\hbar,x_i-x_j)\phi(-q,x_j-x_i)s_{ij}
+\phi(q,x_i-x_j)\phi(-q,x_i-x_j)=
\\ \\
=\phi(\hbar,x_i-x_j)\phi(\hbar,x_j-x_i)-\phi(q,x_i-x_j)\phi(q,x_j-x_i)=\wp(\hbar)-\wp(q).
\end{array}
\eq
In the first equality we used $s_{ij}^2=1$, then (\ref{antiphi}), and (\ref{a10}) in the last equality.

The relation (\ref{uni2}) follows form (\ref{uni}) and (\ref{b25}), (\ref{b26}), (\ref{operrel}). Indeed,
\beq\label{b27b}
 \hat{t}_i R_{ij}^{\hbar}(q)  R_{ji}^{\hbar}(-q)\hat{t_i}=  \hat{t}_i R_{ij}^{\hbar}(q) (\hat{t}_i)^2 R_{ji}^{\hbar}(-q)\hat{t_i}=
 -\tilde{R}_{ij}^{\hbar}(q)  \tilde{R}_{ji}^{-\hbar}(q) .
\eq

Next, let us prove (\ref{uniK}):
\beq\label{b27d}
\begin{array}{ccc}
  K^\hbar_i(q) K^\hbar_i(-q)=\left(v(q,x_i)-v(\hbar,x_i)\hat{t}_i\right)\left(v(-q,x_i)-v(\hbar,x_i)\hat{t}_i\right)=
  \\ \\
 =v(q,x_i)v(-q,x_i)-v(q,x_i)v(\hbar,x_i)\hat{t}_i-v(\hbar,x_i)v(-q,-x_i)\hat{t}_i+v(\hbar,x_i)v(\hbar,-x_i)=
  \\ \\
  \displaystyle
  =v(q,x_i)v(-q,x_i)+v(\hbar,x_i)v(\hbar,-x_i)=
\end{array}
\eq
$$
\begin{array}{ccc}
    \displaystyle
    =\sum_{a=0}^3 \bar\nu^2_a\wp(x_i+\omega_a)-\sum_{a=0}^3 \nu^2_a\wp(q+\omega_a)+\sum_{a=0}^3 \nu^2_a\wp(\hbar+\omega_a)-\sum_{a=0}^3 \bar\nu^2_a\wp(x_i+\omega_a)=
       \\
       \displaystyle
       =\sum_{a=0}^3 \nu^2_a\left(\wp(\hbar+\omega_a)-\wp(q+\omega_a) \right).
\end{array}
$$
The antisymmetry property (\ref{antiphi}) was used in the second equality,
and in the third one equality we used (\ref{q419a}) and (\ref{q415}).

In order to show (\ref{unitildeK}) one should use the relations (\ref{operrel}) and (\ref{b27}):
\beq\label{b27c}
K^\hbar_i(q)  K^\hbar_i(-q)=K^\hbar_i(q) (\hat{t}_i)^2 K^\hbar_i(-q)=-\tilde{K}^\hbar_i(q) \tilde{K}_i^{-\hbar}(q)\,.
\eq
\end{proof}

Let us introduce the notations $F_{ij}^z(q)$ (c.f. (\ref{q03})) and  $\tilde{F}_{ij}^z(q)$ for
the corresponding partial derivatives of $R_{ij}^z(q)$ and  $\tilde{R}_{ij}^z(q)$:
\beq\label{b271}
F_{ij}^z(q):=\partial_{q}R_{ij}^z(q) \qquad
\tilde{F}_{ij}^z(q):=\partial_{q}\tilde{R}_{ij}^z(q)\,,
\eq
and denote by
$Y^\hbar_i(q)$ and $\tilde{Y}^\hbar_i(q)$ the corresponding partial derivatives of the reflection operators:
\beq\label{b272}
Y^\hbar_i(q):=\partial_q K^\hbar_i(q|\bar\nu) \qquad
\tilde{Y}^\hbar_i(q):=\partial_q \tilde{K}^\hbar_i(q|\bar\nu)\,.
\eq
\begin{corollary}
 The following identities hold true :
 \beq\label{b273}
R_{ij}^{\hbar}(q)   F_{ji}^{\hbar}(-q)- F_{ij}^{\hbar}(q) R_{ji}^{\hbar}(-q)=\wp'(q)\,,
 \eq
 \beq\label{b274}
\tilde{R}_{ij}^{\hbar}(q) \tilde{F}_{ji}^{-\hbar}(q)+ \tilde{F}_{ij}^{\hbar}(q) \tilde{R}_{ji}^{-\hbar}(q)=\wp'(q)\,,
 \eq
  \begin{equation}\label{b275}
K^\hbar_i(q)Y^\hbar_i(-q)-  Y^\hbar_i(q) K^\hbar_i(-q) =\sum_{a=0}^3 \nu^2_a \wp'(q+\omega_a)\,,
 \end{equation}
 \beq\label{b276}
 \tilde{Y}^\hbar_i(q) \tilde{K}^{-\hbar}_i(q)+\tilde{K}^\hbar_i(q) \tilde{Y}^{-\hbar}_i(q)=\sum_{a=0}^3 \nu^2_a \wp'(q+\omega_a)\,.
 \eq
%.
\end{corollary}
\begin{proof}
 Taking the partial derivative with respect to $q$ in (\ref{uni}),
   (\ref{uni2}),  (\ref{uniK}),  (\ref{unitildeK}), we obtain  (\ref{b273}),
     (\ref{b274}),   (\ref{b275}),    (\ref{b276}).
 \end{proof}

\subsection{\texorpdfstring{${\rm BC}_n$}{BCn} associative Yang-Baxter equations with parameters}

The construction presented below is motivated by
the Kirillov's ${\rm B}$-type associative Yang-Baxter equation, see Section 5 in \cite{Kir}.
The associative Yang-Baxter equation corresponding to ${\rm A}$-type root systems (\ref{q2}) in the absence of parameters takes the following simple form:
\beq\label{claybe1}
\begin{array}{c}
r_{ij} r_{jk}=r_{ik}r_{ij}+r_{jk}r_{ik} \qquad \text{ for distinct } i,j,k.
\end{array}
\eq
In order to define the  associative classical Yang-Baxter
algebra of type ${\rm A}_n$ generated by the elements $r_{ij}$ for $1\le i\neq j\le n $,
 we extend (\ref{claybe1}) with the antisymmetry condition and the commutativity relation:
\beq\label{claybe2}
\begin{array}{c}
{\displaystyle
r_{ij} r_{kl}=r_{kl}r_{ij} \qquad \text{ for distinct } i,j,k,l, }
\\ \\
{\displaystyle r_{ij}=-r_{ji}}\,.
\end{array}
\eq
The equation (\ref{q2}) with parameters was suggested by A. Polishchuk in \cite{Pol}.
 The relation (\ref{claybe1}) looks as half of
 the classical Yang-Baxter equation. From (\ref{claybe1}) and (\ref{claybe2})
 the classical Yang-Baxter equation follows:
\beq\label{clYB}
[r_{12}, r_{13}]+[r_{12},r_{23}]+[r_{13},r_{23}]=0.
\eq
The classical associative Yang-Baxter algebra of type ${\rm B}_n$ was proposed in \cite{Kir}.
 It is generated by the elements $r_{ij}$, $\tilde{r}_{ij}$ and $y_i$
  for $1\le i\neq j\le n $ satisfying the set of defining relations:
\begin{itemize}
\item symmetry and antisymmetry relations for $i\neq j$:
\beq\label{baybe1}
r_{ij}=-r_{ji}\,, \qquad \tilde{r}_{ij}=\tilde{r}_{ji} \,,
\eq
\item permutation relations for distinct  $i, j, k, l$:
\beq\label{baybe2}
\begin{array}{c}
{\displaystyle y_i y_j=y_j y_i\,, \qquad y_{i}r_{kl}=r_{kl}y_i\,, \qquad y_{i}r_{kl}=r_{kl}y_i\,,}
\\ \\
{\displaystyle r_{ij}r_{kl}=r_{kl}r_{ij},\qquad \tilde{r}_{ij}r_{kl}=r_{kl}\tilde{r}_{ij},\qquad  \tilde{r}_{ij}\tilde{r}_{kl}=\tilde{r}_{kl}\tilde{r}_{ij}\,,}
\end{array}
\eq
\item three term relations for distinct $i,j,k$:
\beq\label{baybe3}
\begin{array}{c}
{\displaystyle
 r_{ij} r_{jk}=r_{ik}r_{ij}+r_{jk}r_{ik}\,,\qquad r_{ij} \tilde{r}_{jk}
 =r_{ik}\tilde{r}_{ij}+\tilde{r}_{jk}r_{ik}\,,
}
\\ \\
{\displaystyle r_{jk}\tilde{r}_{ik}=\tilde r_{ij}r_{jk}
+\tilde{r}_{ik}\tilde{r}_{ij}\,,\qquad \tilde{r}_{jk}r_{ik}
=\tilde{r}_{ij}\tilde{r}_{jk}+\tilde{r}_{ik}\tilde{r}_{ij},
}
\end{array}
\eq
\item four term relation for $i\neq j$:
\beq\label{baybe4}
r_{ij}y_j=y_ir_{ij}+\tilde{r}_{ij}y_i+y_j\tilde{r}_{ij}\,.
\eq
\end{itemize}

\begin{proposition}
The representation of the classical associative Yang-Baxter algebra of type ${\rm B}_n$ can be realized by the following operators:
\beq\label{repaybe}
r_{ij}:=\phi(q_i-q_j,x_i-x_j)\hat{s}_{ij},\qquad \tilde{r}_{ij}:=\phi(q_i+q_j,x_i+x_j)\hat{s}_{ij}\hat{t}_i\hat{t}_j,\qquad y_i:=v(q_i,x_i)\hat{t}_i\,.
\eq
\end{proposition}
\begin{proof}
 The proof is straightforward, one should use the definitions (\ref{b24}), (\ref{tacts}) and move all operators $\hat{s}_{ij}$ and $\hat{t}_{i}$ to the left using relations (\ref{operrel}). The relation (\ref{baybe1}) follows from (\ref{antiphi}). The relations (\ref{baybe3}) follow from the Fay identity (\ref{q114}), (\ref{antiphi}) and the equality in the symmetric group  $s_{ij}s_{jk}=s_{ik}s_{ij}=s_{jk}s_{ik}$.  The four term relation (\ref{baybe4})  follows form (\ref{hident1}).
\end{proof}

Below we formulate a deformation of the relations (\ref{baybe1}), (\ref{baybe2}), (\ref{baybe3}) and (\ref{baybe4})
in the presence of parameters.
\begin{proposition} {\bf The ${\rm BC}_n$ associative Yang-Baxter equation with parameters:}
the following set of relations for the operators (\ref{SUope}), (\ref{Rtilde}), (\ref{Ktilde}) hold true:
 \begin{itemize}
 \item Symmetry and antisymmetry relations:
  \beq\label{b28}
 \tilde{R}^z_{ij}(q)=\tilde{R}^z_{ji}(q)\,,
 \eq
 \beq\label{b29}
 R^z_{ij}(q)=-R_{ji}^{-z}(-q)\,,
 \eq
 \item Permutation relation for distinct $i,j,k,l$:
 \beq\label{comm1}
  \tilde{K}^z_i(q_i) \tilde{K}^w_j(q_j)=  \tilde{K}^w_j(q_j)\tilde{K}_i^z(q_i)\,,
  \qquad
  \tilde{K}^z_i(q_i)  R^w_{kl}(q_{kl})= R^w_{kl}(q_{kl})\tilde{K}_i^z(q_i)\,,
 \eq
 \beq\label{comm2}
   \tilde{K}^z_i(q_i)  \tilde{R}^w_{kl}(q_{kl})=
   \tilde{R}^w_{kl}(q_{kl})\tilde{K}_i^z(q_i)\,,
   \qquad  R^z_{ij}(q_{ij}) R^w_{kl}(q_{kl})=R^w_{kl}(q_{kl}) R^z_{ij}(q_{ij})\,,
 \eq
 \beq\label{comm3}
 R^z_{ij}(q_{ij}) \tilde{R}^w_{kl}(q_{kl})=\tilde{R}^w_{kl}(q_{kl}) R^z_{ij}(q_{ij})\,,
\qquad
 \tilde{R}^z_{ij}(q_{ij}) \tilde{R}^w_{kl}(q_{kl})=\tilde{R}^w_{kl}(q_{kl})
 \tilde{R}^z_{ij}(q_{ij})\,,
 \eq
  \item Three term relations:
 \beq\label{AYBE}
 R_{ij}^z(q_i-q_j) R_{jk}^w(q_j-q_k)=R_{ik}^w(q_i-q_k)R_{ij}^{z-w}(q_i-q_j)+
 R_{jk}^{w-z}(q_j-q_k)R_{ik}^z(q_i-q_k)\,,
 \eq
 \beq\label{AYBE1}
 R_{ij}^z(q_i-q_j) \tilde{R}_{jk}^w(q_j+q_k)=
 \tilde{R}_{ik}^w(q_i+q_k)R_{ij}^{z-w}(q_i-q_j)
 +\tilde{R}_{jk}^{w-z}(q_j+q_k)\tilde{R}_{ik}^z(q_i+q_k)\,,
 \eq
 \beq\label{AYBE2}
 R_{jk}^{z}(q_j-q_k)\tilde{R}_{ik}^{w}(q_i+q_k)=\tilde{R}_{ij}^{w}(q_i+q_j) R_{jk}^{z-w}(q_j-q_k)+\tilde{R}_{ik}^{w-z}(q_i+q_k)\tilde{R}_{ij}^{z}(q_i+q_j)\,,
 \eq
 \beq\label{AYBE3}
 \tilde{R}_{jk}^{z}(q_j+q_k) R_{ik}^{w}(q_i-q_k)=\tilde{R}_{ij}^{w}(q_i+q_j) \tilde{R}_{jk}^{z-w}(q_j+q_k)+R_{ik}^{w-z}(q_i-q_k)\tilde{R}_{ij}^{z}(q_i+q_j)\,,
 \eq
 \item The four term relation (generalization of (\ref{hident1})):
 \beq\label{Fourrel}
 R_{ij}^{w+z}(q_i-q_j)\tilde{K}^w_j(q_j)=\tilde{K}^w_i(q_i)R_{ij}^{z-w}(q_i-q_j)+
 \tilde{R}_{ij}^{w-z}(q_i+q_j)\tilde{K}^z_i(q_i)+
 \tilde{K}_j^{-z}(q_j)\tilde{R}_{ij}^{w+z}(q_i+q_j)\,.
 \eq
  \end{itemize}
\end{proposition}
\begin{proof}
The relations (\ref{b28}) and (\ref{b29}) follows from (\ref{antiphi}). The relations (\ref{comm1}), (\ref{comm2}) and (\ref{comm3}) are valid since $t_i$ , $\hat{t}_j$ and $s_{kl}$  commute for distinct $i,j,k,l$ and do not change any arguments in the corresponding functions $v(z,x_i)$, $v(w,x_j)$ and $\phi(w,x_k\pm x_l)$. Also $s_{ij}s_{kl}=s_{kl}s_{ij}$ and $s_{kl}\phi(z,x_i\pm x_j)=\phi(z,x_i\pm x_j)s_{kl}$.

Let us prove (\ref{AYBE}). We use the
following short notations $q_{ij}=q_i-q_j$ and $x_{ij}=x_i-x_j$ and calculate each summand in (\ref{AYBE}):
 \beq
 \begin{array}{ccc}
 R_{ij}^z(q_i-q_j) R_{jk}^w(q_j-q_k)=\left(\phi(z,x_{ij})-\phi(q_{ij},x_{ij})s_{ij}\right)\left(\phi(w,x_{jk})-\phi(q_{jk},x_{jk})s_{jk}\right)=
 \\
 \\
 =\phi(z,x_{ij})\phi(w,x_{jk})\uuline{-\phi(z,x_{ij})\phi(q_{jk},x_{jk})s_{jk}}
 \\ \\
 \uwave{-\phi(q_{ij},x_{ij})\phi(w,x_{ik})s_{ij}}+\phi(q_{ij},x_{ij})\phi(q_{jk},x_{ik})s_{ij}s_{jk}\,,
 \end{array}
  \eq
   \beq
 \begin{array}{ccc}
 R_{ik}^w(q_i-q_k) R_{ij}^{z-w}(q_i-q_j)=\left(\phi(w,x_{ik})-
 \phi(q_{ik},x_{ik})s_{ik}\right)\left(\phi(z-w,x_{ij})-\phi(q_{ij},x_{ij})s_{ij}\right)=
 \\
 \\
 =\phi(w,x_{ik})\phi(z-w,x_{ij})\uwave{-\phi(w,x_{ik})\phi(q_{ij},x_{ij})s_{ij}}
 \\ \\
 \underline{-\phi(q_{ik},x_{ik})\phi(z-w,x_{kj})s_{ik}}+
 \phi(q_{ik},x_{ik})\phi(q_{ij},x_{kj})s_{ik}s_{ij}\,,
 \end{array}
  \eq
 \beq
 \begin{array}{ccc}
 R_{jk}^{w-z}(q_j-q_k) R_{ik}^{z}(q_i-q_k)=\left(\phi(w-z,x_{jk})-
 \phi(q_{jk},x_{jk})s_{jk}\right)\left(\phi(z,x_{ik})-
 \phi(q_{ik},x_{ik})s_{ik}\right)=
 \\
 \\
 =\phi(w-z,x_{jk})\phi(z,x_{ik})\underline{-\phi(w-z,x_{jk})\phi(q_{ik},x_{ik})s_{ik}}
 \\ \\
 \uuline{-\phi(q_{jk},x_{jk})\phi(z,x_{ij})s_{jk}}+\phi(q_{jk},x_{jk})\phi(q_{ik},x_{ij})s_{jk}s_{ik}\,.
 \end{array}
  \eq
  The corresponding underlined terms are cancelled identically due to the antisymmetry property (\ref{antiphi}).
   The rest of the terms vanish due to the Fay identity ({\ref{a07}}) written as
  \beq
  \phi(z,x_{ij})\phi(w,x_{jk})-\phi(w,x_{ik})\phi(z-w,x_{ij})-\phi(w-z,x_{jk})\phi(z,x_{ik})=0
  \eq
  and
  \beq
  \begin{array}{ccc}
  \phi(q_{ij},x_{ij})\phi(q_{jk},x_{ik})-
  \phi(q_{ik},x_{ik})\phi(q_{ij},x_{kj})-\phi(q_{jk},x_{jk})\phi(q_{ik},x_{ij})=
  \\ \ \\
  \,
    =\phi(q_{ij},x_{ij})\phi(q_{jk},x_{ik})+
    \phi(q_{ik},x_{ik})\phi(q_{ji},x_{jk})-\phi(q_{jk},x_{jk})\phi(q_{ik},x_{ij})=0
  \end{array}
  \eq
  and the following equality in the symmetric group $S_N$: $s_{ij}s_{jk}=s_{ik}s_{ij}=s_{jk}s_{ik}$.

  To prove (\ref{AYBE1}) we conjugate (\ref{AYBE}) by the operator $\hat{t}_j$ and use that $\hat{t}_j^2=1$ and $\hat{t}_j$ commutes with $R_{ik}^\hbar(q_i-q_k)$:
   \beq\label{b57}
 \hat{t}_j R_{ij}^z(q_i-q_j) \hat{t}_j^2 R_{jk}^w(q_j-q_k)\hat{t}_j=R_{ik}^w(q_i-q_k)\hat{t}_j R_{ij}^{z-w}(q_i-q_j)\hat{t}_j+
 \hat{t}_jR_{jk}^{w-z}(q_j-q_k)\hat{t}_j R_{ik}^z(q_i-q_k)\,,
 \eq
  then using (\ref{b25}) we rewrite (\ref{b57}) as
   \beq\label{b58}
 \tilde{R}_{ij}^z(q_i-q_j) \tilde{R}_{jk}^w(q_j-q_k)=R_{ik}^w(q_i-q_k) \tilde{R}_{ij}^{z-w}(q_i-q_j)+
 \tilde{R}_{jk}^{w-z}(q_j-q_k) R_{ik}^z(q_i-q_k)\,.
 \eq
 By replacing $q_k\rightarrow-q_k$ we obtain (\ref{AYBE1}). In the same manner one can prove (\ref{AYBE2}) and (\ref{AYBE3}) conjugating (\ref{AYBE}) by $\hat{t}_k$ and $\hat{t}_i$, respectively, and using (\ref{b26}).

  To prove (\ref{Fourrel}) we calculate each summand moving all operators $\hat{s}_{ij}$, $t_i$ and $t_j$ to the right and using (\ref{operrel}):
  \begin{eqnarray}
R_{ij}^{w+z}(q_i-q_j)\tilde{K}^w_j(q_j)=&\phi(w+z,x_i-x_j)v(w,x_j)\underline{-\phi(q_i-q_j,x_i-x_j)v(w,x_i)\hat{s}_{ij}}\\\notag
&\uuline{-\phi(z+w,x_i-x_j)v(q_j,x_j)\hat{t}_j}+\phi(q_i-q_j,x_i-x_j)v(q_j,x_i)\hat{s}_{ij}\hat{t}_j\,,
\end{eqnarray}
    \begin{eqnarray}
  \tilde{K}^w_i(q_i)R_{ij}^{z-w}(q_i-q_j)=&v(w,x_i)\phi(z-w,x_i-x_j)\underline{-v(w,x_i)\phi(q_i-q_j,x_i-x_j)\hat{s}_{ij}}\\\notag
&\uwave{-v(q_i,x_i)\phi(z-w,-x_i-x_j)\hat{t}_i}+v(q_i,x_i)\phi(q_i-q_j,-x_i-x_j)\hat{s}_{ij}\hat{t}_j\,,
\end{eqnarray}
    \begin{eqnarray}
 \tilde{R}_{ij}^{w-z}(q_i+q_j)\tilde{K}^z_i(q_i)=&\phi(w-z,x_i+x_j)v(z,x_i)\dashuline{-\phi(q_i+q_j,x_i+x_j)v(z,-x_j)\hat{s}_{ij}\hat{t}_i\hat{t}_j}\\\notag
&\uwave{-\phi(w-z,x_i+x_j)v(q_i,x_i)\hat{t}_i}+\phi(q_i+q_j,x_i+x_j)v(q_i,-x_j)\hat{s}_{ij}\hat{t}_j\,,
\end{eqnarray}
    \begin{eqnarray}
\tilde{K}_j^{-z}(q_j)\tilde{R}_{ij}^{w+z}(q_i+q_j)=&v(-z,x_j)\phi(w+z,x_i+x_j)\dashuline{-v(-z,x_j)\phi(q_i+q_j,x_i+x_j)\hat{s}_{ij}\hat{t}_i\hat{t}_j}\\\notag
&\uuline{-v(q_j,x_j)\phi(w+z,x_i-x_j)\hat{t}_j}+v(q_j,x_j)\phi(q_i+q_j,x_i-x_j)\hat{s}_{ij}\hat{t}_j\,.
\end{eqnarray}
The corresponding underlined terms are cancelled identically due (\ref{antiphi}). To finish the proof one should check the identities:
    \begin{eqnarray}\label{b55}
v(w,x_j)\phi(w+z,x_i-x_j)=v(w,x_i)\phi(z-w,x_i-x_j)&+&\\\notag
+v(z,x_i)\phi(w-z,x_i+x_j)&+&v(-z,x_j)\phi(w+z,x_i+x_j)\,,
\end{eqnarray}
and
    \begin{eqnarray}\label{b56}
v(q_j,x_i)\phi(q_i-q_j,x_i-x_j)=v(q_i,x_i)\phi(q_i-q_j,-x_i-x_j)&+&\\
\notag +v(q_i,-x_j)\phi(q_i+q_j,x_i+x_j)&+&v(q_j,x_j)\phi(q_i+q_j,x_i-x_j)\,.
\end{eqnarray}
The identities (\ref{b55}) and (\ref{b56}) follows from (\ref{hident1}).
\end{proof}

 Below we formulate some relations which are used for the proof of the Lax equation in the next Section.
\begin{corollary}
The following identities  hold true:
  \beq\label{q05a}
  \begin{array}{c}
  \displaystyle{
{ R}^z_{ij}{ F}^{z}_{jk}
-{ F}^z_{ij}{ R}^{z}_{jk}=
{ F}^{0}_{jk} R_{ik}^z-R_{ik}^z{ F}^{0}_{ij}\,,
 }
 \end{array}
 \eq
\beq
 \begin{array}{c}\label{q5131}
 	\displaystyle{
 		R^z_{ij}{\tilde F}^{z}_{jk}
 		-{F}^z_{ij}{\tilde R}^{z}_{jk}=
 		{\tilde F}^{0}_{jk}{\tilde R}^{z}_{ik}-{\tilde R}^{z}_{ik}F^{0}_{ij}\,,
 	}
 \end{array}
 \eq
 \beq
 \begin{array}{c}\label{q5132}
 	\displaystyle{
 		{\tilde R}^z_{ij}F^{-z}_{jk}
 		+{\tilde F}^z_{ij}R^{-z}_{jk}=
 		F^{0}_{jk}{\tilde R_{ik}}^z-{\tilde R_{ik}}^z{\tilde F}^{0}_{ij}\,,
 	}
 \end{array}
 \eq
  \beq\label{q513}
  \begin{array}{c}
  \displaystyle{
{\tilde R}^z_{ij}{\tilde F}^{-z}_{jk}
+{\tilde F}^z_{ij}{\tilde R}^{-z}_{jk}=
{\tilde F}^{0}_{jk} R_{ik}^z-R_{ik}^z{\tilde F}^{0}_{ij}\,,
 }
 \end{array}
 \eq
and
\beq\label{b301}
\tilde Y_j^0{\tilde R}_{ij}^{z}-{\tilde R}_{ij}^{z}\tilde Y_i^0=-2
{F}_{ij}^{z}\tilde{K}_j^{z}+2\tilde{K}_i^z
{F}_{ij}^{-z}+{R}_{ij}^{z}\tilde Y_j^{z}+\tilde Y_i^z{R}_{ij}^{-z}\,,
\eq
 \beq\label{b30}
\tilde{Y}_j^0 R_{ij}^{z}-R_{ij}^{z}\tilde{Y}_i^0=2
\tilde{F}_{ij}^{z}\tilde{K}_j^{-z}+2\tilde{K}_i^z\tilde{F}_{ij}^{-z}+
\tilde{R}_{ij}^{z}\tilde{Y}_j^{-z}+\tilde{Y}_i^z\tilde{R}_{ij}^{-z}\,,
\eq
where we use the short notations  $R_{ij}^\mu:=R_{ij}^\mu(q_i-q_j)$, $\tilde{R}_{ij}^\mu:=\tilde{R}_{ij}^\mu(q_i+q_j)$, $\tilde{K}_i^z:=\tilde{K}_i^z(q_i)$, $F_{ij}^z:=F_{ij}^\mu(q_i-q_j)$ and $\tilde{F}_{ij}^z:=\tilde{F}_{ij}^z(q_i+q_j)$, $\tilde{Y}_i^z:=\tilde{Y}_i^z(q_i)$.
\end{corollary}
\begin{proof}
The identities (\ref{q05a}), (\ref{q5131}) are obtained from
 (\ref{AYBE}), (\ref{AYBE1}) by differentiating with respect to $q_j$ and substitution $w:=z$.
In the same manner
the identities (\ref{q5132}), (\ref{q513}) follow from (\ref{AYBE2}), (\ref{AYBE3}) by differentiating
with respect to $q_j$ and substitutions $z:=0$ and $w:=z$.

In order to prove (\ref{b301}) we apply the operator $(\partial_{q_i}-\partial_{q_j})$ to (\ref{Fourrel}):
\beq\label{b31111}
2F_{ij}^{w+z}\tilde{K}_j^w-R_{ij}^{w+z}\tilde{Y}_j^w=\tilde{Y}_i^w R_{ij}^{z-w}+2 \tilde{K}_i^w F_{ij}^{z-w}+\tilde{R}_{ij}^{w-z}\tilde{Y}_i^z-\tilde{Y}_j^{-z}\tilde{R}_{ij}^{w+z}
\eq
and put in $z:=0$ (\ref{b31111}).
By changing the indices $i\leftrightarrow j$ and the variables $q_i\leftrightarrow q_j$, $z\leftrightarrow -z$ in (\ref{Fourrel}), and using the  properties (\ref{b28}),  (\ref{b29}), we obtain:
\beq\label{b31}
\tilde{K}_j^w R_{ij}^{w+z}=R_{ij}^{-w+z}\tilde{K}_i^w+\tilde{R}^{w+z}_{ij}\tilde{K}_j^{-z}+\tilde{K}_i^z \tilde{R}_{ij}^{w-z}\,.
\eq
Next, by applying the operator $(\partial_{q_i}+\partial_{q_j})$ to (\ref{b31}) one gets:
\beq\label{b32}
\tilde{Y}_j^w R_{ij}^{w+z}=R_{ij}^{-w+z}\tilde{Y}_i^w+2 \tilde{F}_{ij}^{w+z}\tilde{K}_j^{-z}+2\tilde{K}_i^z\tilde{F}_{ij}^{w-z}+\tilde{R}_{ij}^{w+z}\tilde{Y}_j^{-z}+\tilde{Y}_i^z\tilde{R}_{ij}^{w-z}\,.
\eq
Finally, we put $w:=0$ and get (\ref{b30}).
\end{proof}

For the proof of the Lax equation in the next Section we also use the identity
 \beq
\begin{array}{c}\label{comm4}
	\displaystyle{
		R^z_{ij}{\tilde R}^{z}_{ji}={\tilde R}^{z}_{ji}R^z_{ij}\,.
	}
\end{array}
\eq
%%%%%%%%%%%%%%%%%%%%%%%%%%%%%%%%%%%%%%%%%%%%%%%%%%%
%%%%%%%%%%%%%%%%%%%%%%%%%%%%%%%%%%%%%%%%%%%%%%%%%%%%

\section{\texorpdfstring{$R$}{R}-matrix-valued generalization of Takasaki's Lax pair}\label{sec5}
\setcounter{equation}{0}
Below we suggest $R$-matrix valued generalization of the Takasaki's Lax pair (\ref{q402})-(\ref{q410}).
That is we deal with $2n\times 2n$ Lax pair which entries are operators acting on functions of
$n$ variables as in (\ref{SUope}) and (\ref{K}).
\subsection{The Lax pair}
  \beq\label{q502}
  \begin{array}{c}
  \displaystyle{
\mL(z)=\mat{\mL^{11}(z)}{\mL^{12}(z)}{\mL^{21}(z)}{\mL^{22}(z)}\,,\qquad
\mM(z)=\mat{\mM^{11}(z)}{\mM^{12}(z)}{\mM^{21}(z)}{\mM^{22}(z)}+\mH\, 1_{2n\times 2n}\,,
 }
 \end{array}
 \eq
where all entries $\mL^{ab}(z)$ and $\mM^{ab}(z)$, $a,b=1,2$ are as follows:
  \beq\label{q503}
  \begin{array}{c}
  \displaystyle{
\mL^{11}_{ij}(z)=\delta_{ij}p_i+g(1-\delta_{ij})R_{ij}^z(q_{ij})\,,
 }
 \end{array}
 \eq
  \beq\label{q504}
  \begin{array}{c}
  \displaystyle{
\mL^{12}_{ij}(z)=\delta_{ij}\tilde{K}^z_i(q_i)+g(1-\delta_{ij})\tilde{R}^z_{ij}(q^+_{ij})\,,
 }
 \end{array}
 \eq
  \beq\label{q505}
  \begin{array}{c}
  \displaystyle{
\mL^{21}_{ij}(z)=-\delta_{ij}\tilde{K}^{-z}_i(q_i)-g(1-\delta_{ij})\tilde{R}^{-z}_{ij}(q^+_{ij})\,,
 }
 \end{array}
 \eq
  \beq\label{q506}
  \begin{array}{c}
  \displaystyle{
\mL^{22}_{ij}(z)=-\delta_{ij}p_i-g(1-\delta_{ij})R^{-z}_{ij}(q_{ij})
 }
 \end{array}
 \eq
and
  \beq\label{q507}
  \begin{array}{c}
  \displaystyle{
\mM^{11}_{ij}(z)=\delta_{ij}A_i+g(1-\delta_{ij})F^z_{ij}(q_{ij})\,,
 }
 \end{array}
 \eq
  \beq\label{q508}
  \begin{array}{c}
  \displaystyle{
\mM^{12}_{ij}(z)=\frac12\,\delta_{ij}\tilde{Y}_i^z(q_i)+g(1-\delta_{ij})\tilde{F}^z_{ij}(q^+_{ij})\,,
 }
 \end{array}
 \eq
  \beq\label{q509}
  \begin{array}{c}
  \displaystyle{
\mM^{21}_{ij}(z)=\frac12\,\delta_{ij}\tilde{Y}^{-z}_i(q_i)+g(1-\delta_{ij})\tilde{F}^{-z}_{ij}(q^+_{ij})\,,
 }
 \end{array}
 \eq
  \beq\label{q510}
  \begin{array}{c}
  \displaystyle{
\mM^{22}_{ij}(z)=\delta_{ij}A_i+g(1-\delta_{ij})F^{-z}_{ij}(q_{ij})\,.
 }
 \end{array}
 \eq
Here
  \beq\label{q511}
  \begin{array}{c}
  \displaystyle{
A_i=-\frac12\,\tilde{Y}_i^0(q_i)-g\sum_{k:k\neq i}^n \Big(F^0_{ik}(q_{ik})+{\ti F}^0_{ik}(q_{ik}^+)\Big)
 }
 \end{array}
 \eq
 and
  \beq\label{q512}
  \begin{array}{c}
  \displaystyle{
\mH=g\sum_{k<l}^n \Big(F^0_{kl}(q_{kl})+{\ti F}^0_{kl}(q_{kl}^+)\Big)+
\frac12\sum\limits_{k=1}^n \tilde{Y}_k^0(q_k)\,.
 }
 \end{array}
 \eq
In the next subsection we prove the following statement.
 \begin{theorem}
 The Lax pair (\ref{q502})-(\ref{q512}) satisfies the Lax equation
  \beq\label{q5124}
  \begin{array}{c}
  \displaystyle{
{\dot \mL}(z)=[\mL(z),\mM(z)]
 }
 \end{array}
 \eq
 and provides the equations of motion for the classical Calogero-Inozemtsev model (\ref{q4011}).
 \end{theorem}

 %%%%%%%%%%%%%%%%%%%%%%%%%%%%%%%%%%%%%%%%%%%%%%%%%%%%%%%%%%%%%%%%%%%%%%%%%%%%%%%%%%%
\subsection{Proof of the Lax equation}
Here we generalize the calculations from the Section \ref{sec32}.
\paragraph{The block 11.}

The form of the equations in this case is similar to (\ref{block11}), but this time
 it is necessary to take into account the additional  term $\mH$ in the diagonal part of $\mM$.
 Consequently, for $a=1,2$ we have
\beq\label{q514}
\begin{array}{c}
  \displaystyle{
\tilde{\mM}^{aa}_{ii}(z):=	\mM^{aa}_{ii}(z)+\mH = A_i+\mH =
\frac12\sum\limits_{k,l\neq i}^n g\Big(F^0_{kl}(q_{kl})+{\ti F}^0_{kl}(q_{kl}^+)\Big)+
\frac12\sum\limits_{k\neq i}^n {\tilde Y}_k^0(q_k)\,,
	}
\end{array}
\eq
where we used the properties
\beq\label{symF}
F_{kl}^0=F_{lk}^0 \qquad \text{and} \qquad \tilde{F}_{kl}^0=\tilde{F}_{lk}^0
\eq
following from (\ref{b28}) and (\ref{b29}).

The diagonal part takes the form:
\beq\label{q422new}
  \begin{array}{c}
  \displaystyle{
\dot{\mL^{11}_{ii}}=\mL^{11}_{ii}{\tilde \mM}^{11}_{ii}-
{\tilde \mM}^{11}_{ii}\mL^{11}_{ii}+\mL^{12}_{ii}\mM^{21}_{ii}-
\mM^{12}_{ii}\mL^{21}_{ii}+
	}
	\\ \ \\
	\displaystyle{
+\sum\limits_{k:k\neq i}^n
\Big(\mL^{12}_{ik}\mM^{21}_{ki}-\mM^{12}_{ik}\mL^{21}_{ki}+
\mL^{11}_{ik}\mM^{11}_{ki}-\mM^{11}_{ik}\mL^{11}_{ki}\Big)\,.
 }
 \end{array}
 \eq
 Plugging the elements of $\mL$- and $\mM$-operators into (\ref{q422new}),
  we obtain the equation of motion (\ref{q4011}):
 \beq\label{q515}
\begin{array}{c}
	\displaystyle{
		\dot{p_{i}}=p_i(A_i+\mH)-(A_i+\mH)p_i+
\frac{1}{2}{\tilde K}^z_i {\tilde Y}^{-z}_i+\frac{1}{2}{\tilde Y}^{z}_i{\tilde K}^{-z}_i+
	}
	\\ \ \\
	\displaystyle{
+g^{2}\sum\limits_{k:k\neq i}^n\Big({\tilde R}^z_{ik}{\tilde F}^{-z}_{ki}+
{\tilde F}^z_{ik}{\tilde R}^{-z}_{ki}+
		R^z_{ik}F^{z}_{ki}-F^z_{ik}R^{z}_{ki}\Big)=
	}
	\\ \ \\
	\displaystyle{
		=\frac{1}{2}\sum\limits_{a=0}^3\nu_a^2\wp'(q_{i}+\om_a)+g^{2}\sum\limits_{k:k\neq i}^n\Big(\wp'(q^{+}_{ik})+\wp'(q_{ik})\Big)\,,
	}
\end{array}
\eq
where we used (\ref{b276}), (\ref{b274}), (\ref{b273}).

For the non-diagonal part we need to show
\beq\label{q424new}
  \begin{array}{c}
  \displaystyle{
\dot{\mL^{11}_{ij}}=\mL^{11}_{ii}\mM^{11}_{ij}-\mM^{11}_{ij}\mL^{11}_{jj}+
\mL^{12}_{ij}\mM^{21}_{jj}-\mM^{12}_{ij}\mL^{21}_{jj}+
\mL^{12}_{ii}\mM^{21}_{ij}-\mM^{12}_{ii}\mL^{21}_{ij}+
 }
 \\ \ \\
 \displaystyle{
+
\mL^{11}_{ij}\tilde{\mM}^{11}_{jj}-\tilde{\mM}^{11}_{ii}\mL^{11}_{ij} +\sum\limits_{k:k\neq i,j}^n\Big(\mL^{12}_{ik}\mM^{21}_{kj}-\mM^{12}_{ik}\mL^{21}_{kj}+
 \mL^{11}_{ik}\mM^{11}_{kj}-\mM^{11}_{ik}\mL^{11}_{kj}\Big)\,.
 }
 \end{array}
 \eq
 Here we use the notation $\tilde{\mM}_{ii}^{aa}$ from (\ref{q514}).
The l.h.s. of (\ref{q424new}) equals $g(\dot{q_i}-\dot{q_j})F^z_{ij}$.
 For the r.h.s. of (\ref{q424new}) we have:
\beq\label{q516}
\begin{array}{c}
	\displaystyle{
		g(p_{i}-p_{j})F^z_{ij}+
		\frac{g}{2}\,{\tilde R}^z_{ij}{\tilde Y}^{-z}_j+
g{\tilde F}^z_{ij}{\tilde K}^{-z}_j+g{\tilde K}^{z}_i{\tilde F}^{-z}_{ij}+
\frac{g}{2}\,{\tilde Y}^{z}_i{\tilde R}^{-z}_{ij}+
	}
    \\ \ \\
    \displaystyle{
    +\frac{g}{2} R^{z}_{ij}\Big(\sum\limits_{k,l\neq j}^n g
    \Big(F^0_{kl}+{\ti F}^0_{kl}\Big)+\sum\limits_{k\neq j}^n {\tilde Y}_k^0\,\Big)-
    \frac{g}{2}\Big(\sum\limits_{k,l\neq i}^n g\Big(F^0_{kl}+{\ti F}^0_{kl}\Big)+
    \sum\limits_{k\neq i}^n {\tilde Y}_k^0\,\Big)R^{z}_{ij}+
    }
    \\ \ \\
    \displaystyle{
		+g^2\sum\limits_{k:k\neq i,j}^n\Big({\tilde R}^z_{ik}{\tilde F}^{-z}_{kj}
			+{\tilde F}^z_{ik}{\tilde R}^{-z}_{kj}
			+R^z_{ik}F^z_{kj}-F^z_{ik}R^z_{kj}\Big) =
	}
    \end{array}
	\eq
 $$
\begin{array}{c}
    \displaystyle{
		=g(p_{i}-p_{j})F^z_{ij}+g\left(
		\frac{1}{2}\,{\tilde R}^z_{ij}{\tilde Y}^{-z}_j+
{\tilde F}^z_{ij}{\tilde K}^{-z}_j+\frac{1}{2}\,{\tilde Y}^{z}_i{\tilde R}^{-z}_{ij}+{\tilde K}^{z}_i{\tilde F}^{-z}_{ij}+\frac12 R^{z}_{ij}{\tilde Y}_i^0-\frac12 {\tilde Y}_j^0R^{z}_{ij}\right)+
	}
    \\ \ \\
    \displaystyle{
    \underline{+g^2R^{z}_{ij}\sum\limits_{k:k\neq i,j}^n
    \Big(F^0_{ik}+{\ti F}^0_{ik}\Big)
    -g^2\sum\limits_{k:k\neq i,j}^n \Big(F^0_{kj}+{\ti F}^0_{kj}\Big)R^{z}_{ij}+}
    }
    \\ \ \\
    \displaystyle{
		\underline{+g^2\sum\limits_{k:k\neq i,j}^n\Big({\tilde R}^z_{ik}{\tilde F}^{-z}_{kj}
			+{\tilde F}^z_{ik}{\tilde R}^{-z}_{kj}
			+R^z_{ik}F^z_{kj}-F^z_{ik}R^z_{kj}\Big)} =
	}
    \\ \ \\
    \displaystyle{
		=g(p_{i}-p_{j})F^z_{ij}\,.
	}
\end{array}
$$
In the first equality we used the commutativity properties (\ref{comm1}), (\ref{comm2}), (\ref{comm3})
 and their consequences
 $F_{kl}^0R_{ij}^z=R_{ij}^zF_{kl}^0$ and $\tilde{F}_{kl}^0R_{ij}^z=R_{ij}^z\tilde{F}_{kl}^0$
 for distinct $i,j,k,l$, so that
the corresponding terms with distinct indices vanish. Also, we used (\ref{symF}).
The terms inside the brackets vanish due to (\ref{b30}).  The underlined terms
are cancelled out due to (\ref{q5132}) and (\ref{q05a}).
\paragraph{The block 12.}
In the diagonal part we need to prove the following equation
\beq\label{q426new}
  \begin{array}{c}
  \displaystyle{
\dot{\mL^{12}_{ii}}=\mL^{11}_{ii}\mM^{12}_{ii}-
\mM^{12}_{ii}\mL^{22}_{ii}+\mL_{ii}^{12}\tilde\mM_{ii}^{22}-\tilde\mM_{ii}^{11}\mL_{ii}^{12}+}
\\ \ \\
\displaystyle{+\sum\limits_{k:k\neq i}^n
\Big(\mL^{11}_{ik}\mM^{12}_{ki}-\mM^{11}_{ik}\mL^{12}_{ki}+
\mL^{12}_{ik}\mM^{22}_{ki}-\mM^{12}_{ik}\mL^{22}_{ki}\Big)\,.
 }
 \end{array}
 \eq
Here we have additional terms $\mL_{ii}^{12}\tilde\mM_{ii}^{22}-\tilde\mM_{ii}^{11}\mL_{ii}^{12}$,
which are absent in (\ref{q426}). However, the equality
$\mL_{ii}^{12}\tilde\mM_{ii}^{22}=\tilde\mM_{ii}^{11}\mL_{ii}^{12}$ holds true in the operator case as well.
This follows from the explicit form of $\tilde\mM_{ii}^{aa}$ given in (\ref{q514}) and (\ref{comm1}), (\ref{comm2}).
Therefore, similarly to (\ref{q427}) we have:
\beq\label{q517}
\begin{array}{c}
	\displaystyle{
		\dot{q_{i}}{\tilde Y}^z_i=p_i{\tilde Y}^z_i+
g^{2}\sum\limits_{k:k\neq i}^n\Big(R^z_{ik}(q_{ik}){\tilde F}^{z}_{ki}(q^{+}_{ik})+{\tilde F}^z_{ik}(q^{+}_{ik})R^{-z}_{ki}(q_{ki})+
	}
	\\ \ \\
	\displaystyle{
		+{\tilde R}^z_{ik}(q^{+}_{ik})F^{-z}_{ki}(q_{ki})-
F^z_{ik}(q_{ik}){\tilde R}^{z}_{ki}(q^{+}_{ki})\Big)=p_i{\tilde Y}^z_i\,,
	}
\end{array}
\eq
where we used the
commutativity\footnote{A sufficient condition is that $\displaystyle\frac{\partial}{\partial q_k}\left(R_{ik}^z(q_i-q_k)\tilde{R}^z_{ik}(q_i+q_k)-\tilde{R}_{ik}^z(q_i+q_k)R^z_{ik}(q_i-q_k)\right)=0$.}
 (\ref{comm4})  and the skew-symmetry (\ref{b29}).

For the non-diagonal part we have the equation:

\beq\label{q428new}
  \begin{array}{c}
  \displaystyle{
\dot{\mL^{12}_{ij}}=\mL^{11}_{ii}\mM^{12}_{ij}-\mM^{12}_{ij}\mL^{22}_{jj}+
\mL^{11}_{ij}\mM^{12}_{jj}-\mM^{11}_{ij}\mL^{12}_{jj}+
\mL^{12}_{ii}\mM^{22}_{ij}-
\mM^{12}_{ii}\mL^{22}_{ij}+\mL^{12}_{ij}\tilde\mM^{22}_{jj}-\tilde{\mM}^{11}_{ii}\mL^{12}_{ij}
 }
 \\ \ \\
 \displaystyle{
 +\sum\limits_{k:k\neq i,j}^n\Big(\mL^{11}_{ik}\mM^{12}_{kj}-\mM^{11}_{ik}\mL^{12}_{kj}+
 \mL^{12}_{ik}\mM^{22}_{kj}-\mM^{12}_{ik}\mL^{22}_{kj}\Big)\,.
 }
 \end{array}
 \eq

Similarly to (\ref{q516}), the l.h.s. of (\ref{q428new}) equals $g(\dot{q_i}+\dot{q_j}){\tilde F}^z_{ij}$.
For the r.h.s. of (\ref{q428new}) we have:
\beq\label{q518}
\begin{array}{c}
	\displaystyle{
		g(p_{i}+p_{j}){\tilde F}^z_{ij}+
		\frac{g}{2}\,{R}^z_{ij}{\tilde Y}^{z}_j-g{F}^z_{ij}{\tilde K}^{z}_j+g{\tilde K}^{z}_i{F}^{-z}_{ij}+\frac{g}{2}\,{\tilde Y}^{z}_i{R}^{-z}_{ij}+
	}
	\\ \ \\
	\displaystyle{
		+\frac g2 {\tilde R}^{z}_{ij}\Big(\sum\limits_{k,l\neq j}^n g
\Big(F^0_{kl}+{\ti F}^0_{kl}\Big)+\sum\limits_{k\neq j}^n {\tilde Y}_k^0\,\Big)-
\frac g2\Big(\sum\limits_{k,l\neq i}^n g\Big(F^0_{kl}+{\ti F}^0_{kl}\Big)
+\sum\limits_{k\neq i}^n {\tilde Y}_k^0\,\Big){\tilde R}^{z}_{ij}+
	}
	\\ \ \\
	\displaystyle{
		+g^2\sum\limits_{k:k\neq i,j}^n\Big(R^z_{ik}{\tilde F}^z_{kj}-F^z_{ik}{\tilde R}^z_{kj}+{\tilde R}^z_{ik}{F}^{-z}_{kj}
		+{\tilde F}^z_{ik}{R}^{-z}_{kj}
		\Big) =
	}
	\end{array}
	\eq
	$$
	\begin{array}{c}
	\displaystyle{
		=g(p_{i}+p_{j}){\tilde F}^z_{ij}+g\left(
\frac{1}{2}\,{R}^z_{ij}{\tilde Y}^{z}_j-{F}^z_{ij}{\tilde K}^{z}_j+
{\tilde K}^{z}_i{F}^{-z}_{ij}+\frac{1}{2}\,{\tilde Y}^{z}_i{R}^{-z}_{ij}
+\frac12 {\tilde R}^{z}_{ij}{\tilde Y}_i^0-\frac12 {\tilde Y}_j^0{\tilde R}^{z}_{ij}\right)
	}
	\\ \ \\
	\displaystyle{
		\underline{+g^2{\tilde R}^{z}_{ij}\sum\limits_{k:k\neq i,j}^n\Big(F^0_{ik}+{\ti F}^0_{ik}\Big)-g^2\sum\limits_{k:k\neq i,j}^n \Big(F^0_{kj}+{\ti F}^0_{kj}\Big){\tilde R}^{z}_{ij}+}
	}
	\\ \ \\
	\displaystyle{
		\underline{+g^2\sum\limits_{k:k\neq i,j}^n
\Big(R^z_{ik}{\tilde F}^z_{kj}-F^z_{ik}{\tilde R}^z_{kj}+{\tilde R}^z_{ik}{F}^{-z}_{kj}
		+{\tilde F}^z_{ik}{R}^{-z}_{kj}\Big)} =
	}
	\\ \ \\
	\displaystyle{
		=g(p_{i}+p_{j}){\tilde F}^z_{ij}\,.
	}
\end{array}
$$
The terms inside the brackets  vanish due to (\ref{b301}).
 The underlined terms are cancelled out due to (\ref{q5131}), (\ref{q5132}).
\paragraph{The block 21 and 22.}
In the R-operator case, (anti)symmetry properties
(\ref{Lanti}), (\ref{Msym}) for $\mL$, $\mM$ remain valid:
\beq\label{Lanti2}
\mL^{11}(z)=-\mL^{22}(-z), \qquad \mL^{12}(z)=-\mL^{21}(-z),
\eq
\beq\label{Msym2}
{\tilde \mM}^{11}(z)={\tilde \mM}^{22}(-z), \qquad \mM^{12}(z)=\mM^{21}(-z).
\eq
The equations
\beq\label{Rblock21}
\dot{\mL}^{21}(z)=\mL^{21}(z){\tilde \mM}^{11}(z)+\mL^{22}(z)\mM^{21}(z)-\mM^{21}(z)\mL^{11}(z)-{\tilde \mM}^{22}(z)\mL^{21}(z),
\eq
\beq\label{Rblock22}
\dot{\mL}^{22}(z)=\mL^{21}(z)\mM^{12}(z)+\mL^{22}(z){\tilde \mM}^{22}(z)-\mM^{21}(z)\mL^{12}(z)-{\tilde \mM}^{22}(z)\mL^{22}(z)
\eq
are correct due to (\ref{q422new}), (\ref{q424new}), (\ref{q426new}), (\ref{q428new}) and (\ref{Lanti2})-(\ref{Msym2}).
%
%%%%%%%%%%%%%%%%%%%%%%%%%%%%%%%%%%%%%%%%%%%%%%%%%%%%%%%%%%%%%%%%%%%%%%%%%%%%%%%%%%%%%%%%%%%%%%%%%
%%%%%%%%%%%%%%%%%%%%%%%%%%%%%%%%%%%%%%%%%%%%%%%%%%%%%%%%%%%%%%%%%%%%%%%%%%%%%%%%%%%%%%%%%%%%%%%%%
%%%%%%%%%%%%%%%%%%%%%%%%%%%%%%%%%%%%%%%%%%%%%%%%%%%%%%%%%%%%%%%%%%%%%%%%%%%%%%%%%%%%%%%%%%%%%%%%%
\section{Representation via the Baxter's 8-vertex \texorpdfstring{$R$}{R}-matrix}\label{sec6}
\setcounter{equation}{0}
The Baxter's $R$-matrix
 has the form \cite{BB}:
 \beq\label{w810}
 \begin{array}{c}
  \displaystyle{
 R_{12}^\hbar(z)
 =\frac12\sum\limits_{k=0}^3\vf_{k}\Big(z,\om_{k}+\frac{\hbar}{2}\Big)\sigma_{4-k}\otimes\sigma_{4-k}\,,
  }
 \end{array}
\eq
 where
 \beq\label{w8101}
 \begin{array}{c}
  \displaystyle{
 \sigma_0=\mats{1}{0}{0}{1}\,,\quad
  \sigma_1=\mats{0}{1}{1}{0}\,,\quad
  \sigma_2=\mats{0}{-\imath}{\imath}{0}\,,\quad
   \sigma_3=\mats{1}{0}{0}{-1}
  }
 \end{array}
\eq
 are the Pauli matrices and
 $\sigma_4=\sigma_0=1_{2\times 2}$.  Equivalently,
 \beq\label{w207}
 \begin{array}{c}
 R_{12}^\hbar(z)=\displaystyle{\frac12}\left(
 \begin{array}{cccc}
 \vf_{00}+\vf_{10} & 0 & 0 & \vf_{01}-\vf_{11}
 \\
 0 & \vf_{00}-\vf_{10} & \vf_{01}+\vf_{11} & 0
  \\
 0 & \vf_{01}+\vf_{11} & \vf_{00}-\vf_{10} & 0
 \\
 \vf_{01}-\vf_{11} & 0 & 0 & \vf_{00}+\vf_{10}
 \end{array}
 \right)\,,
  \end{array}
 \eq
  where
 \beq\label{w208}
 \begin{array}{c}
  \displaystyle{
 \vf_{00}=\phi\Big(z,\frac{\hbar}{2}\Big)=\vf_{0}\Big(z,\om_{0}+\frac{\hbar}{2}\Big)\,,\quad
 \vf_{10}=\phi\Big(z,\frac{1+\hbar}{2}\Big)=\vf_{1}\Big(z,\om_{1}+\frac{\hbar}{2}\Big)\,,
 }
 \\ \ \\
  \displaystyle{
 \vf_{01}=e^{\pi\imath z}\phi\Big(z,\frac{\tau+\hbar}{2}\Big)=\vf_{3}\Big(z,\om_{3}+\frac{\hbar}{2}\Big)\,,\quad
 \vf_{11}=e^{\pi\imath z}\phi\Big(z,\frac{1+\tau+\hbar}{2}\Big)=\vf_{2}\Big(z,\om_{2}+\frac{\hbar}{2}\Big)\,.
 }
  \end{array}
 \eq
Define also the $K$-matrix as
 \beq\label{w811}
 \begin{array}{c}
  \displaystyle{
 {\tilde K}^\hbar(z)=\sum\limits_{k=0}^3\nu_{k}
 e^{-2\pi\imath\hbar\p_\tau\om_{k}}\vf_{k}(z-\om_{k},\hbar+\om_{k})\sigma_{4-k}=
}
\\ \ \\
  \displaystyle{
 =
 \sum\limits_{k=0}^3\nu_{k}
 e^{2\pi\imath (z+\hbar+\om_k)\p_\tau\om_{k}}\phi(z+\om_{k},\hbar+\om_{k})\sigma_{4-k}\,,
  }
 \end{array}
\eq
 where $\nu_4=\nu_0$ is assumed, and $\p_\tau\om_{k}$ is equal to either $0$ (for $\om_0$, $\om_1$) or
 $1/2$ (for $\om_2$, $\om_3$).
  Up to simple redefinitions it coincides with the known $K$-matrix from \cite{IK,KH2}.
  It satisfies the property (\ref{unitildeK}) that can be deduced from (\ref{a10}) and the
  quasi-periodic behaviour (\ref{a0621})-(\ref{a0624}).
  Notice also that ${\tilde K}^\hbar(z)={\tilde K}^z(\hbar)$, that is ${K}^\hbar(z)={\tilde K}^\hbar(z)$ in this description.
 \begin{predl}
 Let $R^\hbar_{ij}(q)$ be the Baxter's $R$-matrix (\ref{w810}).
 Then the $K$-matrix (\ref{w811}) satisfies the relations (\ref{Fourrel}) and (\ref{b30}) with
 $R^\hbar_{ij}=R^\hbar_{ij}(q_i-q_j)$ and ${\ti R}^\hbar_{ij}= R^\hbar_{ij}(q_i+q_j)$.
 \end{predl}
 The proof is by straightforward calculation. One should use the set of identities, which are similar
 to those in (\ref{hident1}).

In fact, all the properties or $R$-matrix and $K$-matrix necessary for the proof of the
$R$-matrix valued Lax equation hold true.
 \begin{predl}
 The Lax equation (\ref{q006}) with the Lax pair (\ref{q502})-(\ref{q512})
 defined through the Baxter's elliptic $R$-matrix (\ref{w810})
 and the $K$-matrix (\ref{w811})
 is equivalent to
 the equations of motion (\ref{q4011}).
 \end{predl}
 Let us write down explicit form for the expression $\mH$ (\ref{q512}):
  \beq\label{q5121}
  \begin{array}{c}
  \displaystyle{
\mH=g\sum_{k<l}^n \Big(F^0_{kl}(q_{k}-q_l)+{\ti F}^0_{kl}(q_{k}+q_l)\Big)
+\frac12\sum\limits_{k=1}^n \tilde{Y}_k^0(q_k)\in{\rm Mat}(2,\mC)^{\otimes n}\,,
 }
 \end{array}
 \eq
 where
  \beq\label{q5122}
  \begin{array}{c}
  \displaystyle{
F^0_{12}(x)={\ti F}^0_{12}(x)=\p_x R_{12}^\hbar(x)\Big|_{\hbar=0}=
  }
  \\ \ \\
  \displaystyle{
=-\frac12\, E_2(x)\sigma_0\otimes\sigma_0+
\frac12\sum\limits_{a=1}^3\vf_a(x,\om_a)\Big(E_1(x+\om_a)-E_1(x)-E_1(\om_a)\Big)
\sigma_{4-a}\otimes\sigma_{4-a}
 }
 \end{array}
 \eq
 and
  \beq\label{q5123}
  \begin{array}{c}
  \displaystyle{
\tilde{Y}^0(x)=\p_x K^\hbar(x)\Big|_{\hbar=0}=
  }
  \\ \ \\
  \displaystyle{
=-\nu_0 E_2(x)\sigma_0-\sum\limits_{a=1}^3\nu_a
\vf_a(x-\om_a,\om_a)\Big(E_1(x+\om_a)-E_1(x)-E_1(\om_a)\Big)\sigma_{4-a}\,.
 }
 \end{array}
 \eq
 %

%%%%%%%%%%%%%%%%%%%%%%%%%%%%%%%%%%%%%%%%%%%%%%%%%%%%%%%%%%%%%%%%%%%%%%%%%%%%%%%%%%%%%%%%%%%%%%%%%
%%%%%%%%%%%%%%%%%%%%%%%%%%%%%%%%%%%%%%%%%%%%%%%%%%%%%%%%%%%%%%%%%%%%%%%%%%%%%%%%%%%%%%%%%%%%%%%%%
\section{Concluding remarks: XYZ long-range spin chain of \texorpdfstring{${\rm BC}_n$}{BCn} type}\label{sec7}
\setcounter{equation}{0}

As was explained in the end of Section \ref{sec2} the classical $R$-matrix-valued Lax pair allows us to
describe a quantum long-range spin chains. Consider some equilibrium position
in the Calogero-Inozemtsev model
  \beq\label{q991}
  \begin{array}{c}
  \displaystyle{
p_i=0\,,\quad q_i=\zeta_i\,,\ i=1,...,n
 }
 \end{array}
 \eq
defined by the equations ${\dot p}_i=0$, that is
\begin{equation}
 \label{q992}
   \displaystyle{
g^2
\sum_{k:k\neq i}^{n} \Big(\wp'(\zeta_i-\zeta_k) + \wp'(\zeta_i+\zeta_k)\Big) +
\frac{1}{2}\sum\limits_{a=0}^3\nu_a^2\wp'(\zeta_i+\om_a)=0\,,\quad i=1,...,n\,.
}
 \end{equation}
To the best of our knowledge, these equilibrium positions are unknown in the general case.
At the same time, it is known that $\zeta_i$ coincide with positions of zeros
of solutions to the Heun equation (in the elliptic form) \cite{GW}.
Consider the restriction of $\mL(z)$ and $\mM(z)$ (\ref{q502})-(\ref{q512}) to
the positions (\ref{q991}):
\begin{equation}
 \label{q993}
   \displaystyle{
\mL'(z)=\mL(z)|_{p_j=0,q_j=\zeta_j}\,,\quad \mM'(z)=\mM(z)|_{q_j=\zeta_j}\,,\quad
{\mH^0}'=\mH^0|_{q_j=\zeta_j}\,.
}
 \end{equation}
Then we come to the quantum Lax equation
  \beq\label{q994}
  \begin{array}{c}
  \displaystyle{
 [{\mH^0}',\mL'(z)]=[\mL'(z),\mM'(z)]\,,
 }
 \end{array}
 \eq
where the quantum Hamiltonian ${\mH^0}'$ is as follows:
  \beq\label{q995}
  \begin{array}{c}
  \displaystyle{
{\mH^0}'=g\sum_{k<l}^n \Big(F^0_{kl}(\zeta_k-\zeta_l)+{\ti F}^0_{kl}(\zeta_k+\zeta_l)\Big)
+\frac12\sum\limits_{k=1}^n \tilde{Y}_k^0(\zeta_k)=
 }
 \\
   \displaystyle{
=-\frac{g}2\sum_{k<l}^n
 \Big(E_2(\zeta_k-\zeta_l)+E_2(\zeta_k+\zeta_l)\Big)\stackrel{k}{\sigma}_0\otimes\stackrel{l}{\sigma}_0+
 }
 \\
   \displaystyle{
+
\frac{g}2\sum_{k<l}^n\sum\limits_{a=1}^3\Big[\vf_a(\zeta_k-\zeta_l,\om_a)\Big(E_1(\zeta_k-\zeta_l+\om_a)-E_1(\zeta_k-\zeta_l)-E_1(\om_a)\Big)+
}
\\
\displaystyle{
+\vf_a(\zeta_k+\zeta_l,\om_a)\Big(E_1(\zeta_k+\zeta_l+\om_a)-E_1(\zeta_k+\zeta_l)-E_1(\om_a)\Big)\Big]
\stackrel{k}\sigma_{4-a}\otimes\stackrel{l}\sigma_{4-a}-
 }
 \\
   \displaystyle{
 -\frac12\sum\limits_{k=1}^n\Big[ \nu_0 E_2(\zeta_k)\stackrel{k}\sigma_0+\sum\limits_{a=1}^3\nu_a
\vf_a(\zeta_k-\om_a,\om_a)\Big(E_1(\zeta_k+\om_a)-E_1(\zeta_k)-E_1(\om_a)\Big)\stackrel{k}\sigma_{4-a}\Big]\,.
}
 \end{array}
 \eq
In fact, the existence of the Lax equation is not sufficient for integrability. However, we expect the model to be integrable,
and we hope to prove this conjecture in future works.
The properties of this model will be studied elsewhere.

%\subsection{Kirillov's $B$-type associative Yang-Baxter equation}

%\paragraph{Quantum Lax pair.}
%In quantum case ... \cite{GrZ18}

%\paragraph{$R$-matrix valued Lax pair for higher Painlev\'e equations.}

%\paragraph{${BC}_n$ version of anisotropic Inozemtsev long-range spin chain.}

%equilibrium positions.. ssylka na Vadima i ssylka ot nego

%\paragraph{Lax pairs for Calogero models in external fields.}
%???

\subsection*{\ \ \ \ \ \ \underline{Main notations:}}
%\addcontentsline{toc}{section}{ Notations}

\ \vspace{-7mm}

{\small{

$n$ -- number of particles in the elliptic Calogero-Moser (\ref{q001}) and the Calogero-Inozemtsev model (\ref{q0401});

$q_i$ -- the $i$-th particle coordinate;

$p_i$ -- the $i$-th particle momenta;

$q_{ij}$ and $q_{ij}^+$ -- short notation for $q_i-q_j$ and $q_i+q_j$, respectively;

$\zeta_i$ -- equilibrium positions, see (\ref{q3402}) and (\ref{q991});

$\nu_a$ for $a=0,1,2,3$ -- complex parameters in the Calogero-Inozemtsev model (\ref{q0401})

$N$ -- rank of ${\rm GL}_N$ $R$-matrix;

$\tau$ -- moduli of elliptic curve $\mC/\Gamma$, ${\rm Im}(\tau)>0$;

$\phi(z,u)$ -- Kronecker elliptic function $\phi(z,u)$ (\ref{a001}), (\ref{a01});

$\wp(z)$ -- the Weiestrass elliptic $\wp$-function;

$v(z,u)$ --elliptic function (\ref{q419}) depending on the parameters $\nu_0,\nu_1,\nu_2,\nu_3$ ;

$\hbar$ -- complex quantum parameter entering the $R$-operator definition;

$z$ -- spectral parameter in Lax matrix, at the same time the quantum parameter in $R$-operators in $R$-matrix Lax pairs;

$R^\hbar(u)$ --- elliptic $R$-operator, satisfying (\ref{q2}), (\ref{q008}), (\ref{q01}) and (\ref{q011}). Depending on the context, it is the Shibukawa-Ueno operator (\ref{SUope}), which acts on the space of scalar functions on $n$-variables, or its finite-dimensional analogue, the Baxter-Belavin elliptic $R$-matrix, acting on $\mC^N\otimes \mC^N$, see (\ref{w810}) for $N=2$;

$R_{ij}^\hbar(u)$ -- $R$-operator, acting nontrivially on variables $x_i$ and $x_j$, or $R$-matrix, i.e. an element of ${\rm End}((\mC^N)^{\otimes n})$ acting on the $i$-th and the $j$-th tensor components of $(\mC^N)^{\otimes n}$;

$\tilde{K}_i^\hbar(u)$ -- the corresponding $K$-operator, satisfying ({\ref{q0014}}) and (\ref{RE}), see (\ref{Ktilde}) and (\ref{w811});

$F_{ij}^\hbar(u)$ -- partial derivative of $R$-operator $\partial_u R_{ij}^\hbar (u)$;

$\tilde{Y}_{i}^\hbar(u)$ -- partial derivative of $K$-operator $\partial_u K_{i}^\hbar (u)$;

$R_{ij}^\hbar$,
$F_{ij}^\hbar$  --short notation for $R_{ij}^\hbar(q_i-q_j)$ and $F_{ij}^\hbar(q_i-q_j)$;

$\tilde K_{ij}^\hbar$,
$\tilde Y_{ij}^\hbar$  --short notation for $\tilde K_{i}^\hbar(q_i)$ and $\tilde Y_{i}^\hbar(q_i)$;

$L(z)$ and $M(z)$ -- Lax pair in the scalar case, see (\ref{q004}) and (\ref{q105}) for the elliptic Calogero-Moser and (\ref{q402}) for the Calogero-Inozemtsev model;

$\mL(z)$ and $\bar\mM(z)$ -- $R$-matrix valued Lax pair, see (\ref{q31})-(\ref{q33}) for the elliptic Calogero-Moser and (\ref{q502}) for the Calogero-Inozemtsev model.

%%%%%%%%%%%%%%%%%%%%%%%%%%%%%%%%%%%%%%%%%%%%%%%%%%%%%%%%%%%%%%%%%%%%%%%%%%%%%%%%%%%%%%%%%%%%%%%%%
\section{Appendix: elliptic functions}\label{secA}
\def\theequation{A.\arabic{equation}}
\setcounter{equation}{0}

\paragraph{Main definitions and properties.}
We mainly deal with the elliptic Kronecker function \cite{Weil}:
 \beq\label{a01}
  \begin{array}{l}
  \displaystyle{
 \phi(z,u)=\frac{\vth'(0)\vth(z+u)}{\vth(z)\vth(u)}=\phi(u,z)\,,\quad
 \res\limits_{z=0}\phi(z,u)=1\,,
  \quad \phi(-z, -u) = -\phi(z, u)\,,
 }
 \end{array}
 \eq
where $\vth(z)$ is the first Jacobi theta-function:
 \beq\label{a02}
 \begin{array}{c}
  \displaystyle{
\vth(z)=\vth(z,\tau)\equiv-\theta{\left[\begin{array}{c}
1/2\\
1/2
\end{array}
\right]}(z|\, \tau )\,,
 }
 \end{array}
 \eq
\beq\label{a03}
 \begin{array}{c}
  \displaystyle{
\theta{\left[\begin{array}{c}
a\\
b
\end{array}
\right]}(z|\, \tau ) =\sum_{j\in \mZ}
\exp\left(2\pi\imath(j+a)^2\frac\tau2+2\pi\imath(j+a)(z+b)\right)\,,\quad {\rm Im}(\tau)>0\,.
}
 \end{array}
 \eq
Here $\tau$ is the moduli of elliptic curve $\mC/(\mZ+\tau\mZ)$, ${\rm Im}(\tau)>0$.
It is sometimes useful to use the Jacobi theta-functions:
\beq\label{8v1}
\begin{array}{c}
\displaystyle{
\theta_1(u|\tau )=\vth(u,\tau)=-i\sum_{k\in \mZ}
(-1)^k q^{(k+\frac{1}{2})^2}e^{\pi i (2k+1)u},
\qquad
\theta_2(u|\tau )=\sum_{k\in \mZ}
q^{(k+\frac{1}{2})^2}e^{\pi i (2k+1)u},}
\\ \\
\displaystyle{
\theta_3(u|\tau )=\sum_{k\in \mZ}
q^{k^2}e^{2\pi i ku},
\qquad
\theta_4(u|\tau )=\sum_{k\in \mZ}
(-1)^kq^{k^2}e^{2\pi i ku},}
\end{array}
\eq
where
$q=e^{\pi i \tau}$ and
\beq\label{8v11}
\begin{array}{c}
\theta_2(u|\tau)=\theta_1(u+\frac{1}{2}|\tau), \! \quad \! \!
\theta_3(u|\tau)=q^{\frac{1}{4}}e^{\pi iu}\theta_1(u+\frac{\tau+1}{2}|\tau), \! \quad \!\!
\theta_4(u|\tau)=-iq^{\frac{1}{4}}e^{\pi iu}\theta_1(u+\frac{\tau}{2}|\tau).
\end{array}
\eq
The infinite product representation is
%Their infinite product representations are also useful:
\beq\label{8v2}
\begin{array}{l}
\displaystyle{
\theta_1(u|\tau )=2q^{\frac{1}{4}} \sin \pi u \prod_{n\geq 1}
(1-q^{2n})(1-q^{2n}e^{2\pi i u})(1-q^{2n}e^{-2\pi i u}),
}
\\ \\
\displaystyle{
\theta_2(u|\tau )=2q^{\frac{1}{4}} \cos \pi u \prod_{n\geq 1}
(1-q^{2n})(1+q^{2n}e^{2\pi i u})(1+q^{2n}e^{-2\pi i u}),
}
\\ \\
\displaystyle{
\theta_3(u|\tau )= \prod_{n\geq 1}
(1-q^{2n})(1+q^{2n-1}e^{2\pi i u})(1+q^{2n-1}e^{-2\pi i u}),
}
\\ \\
\displaystyle{
\theta_4(u|\tau )= \prod_{n\geq 1}
(1-q^{2n})(1-q^{2n-1}e^{2\pi i u})(1-q^{2n-1}e^{-2\pi i u}).
}
\end{array}
\eq
The partial derivative $f(z,u) = \partial_u \vf(z,u)$ equals
\beq\label{a04}
\begin{array}{c} \displaystyle{
    f(z, u) = \phi(z, u)(E_1(z + u) - E_1(u)), \qquad f(-z, -u) = f(z, u)\,,
}\end{array}\eq
where $E_1(z)$ the first
  Eisenstein function. The first and the second Eisenstein functions are defined as follows:
\beq\label{a05}
\begin{array}{c} \displaystyle{
    E_1(z)=\frac{\vth'(z)}{\vth(z)}=\zeta(z)+\frac{z}{3}\frac{\vth'''(0)}{\vth'(0)}\,,
    \quad
    E_2(z) = - \partial_z E_1(z) = \wp(z) - \frac{\vartheta'''(0) }{3\vartheta'(0)}\,,
}\end{array}\eq
\beq\label{a06}
\begin{array}{c}
 \displaystyle{
    E_1(- z) = -E_1(z)\,, \quad E_2(-z) = E_2(z)\,,
}\end{array}
\eq
where $\wp(z)$ and $\zeta(z)$ are the Weierstrass functions.
Due to the following behavior of $\phi(z,q)$ near $z=0$
  \beq\label{a061}
  \begin{array}{l}
  \displaystyle{
\phi(z,q)=z^{-1}+E_1(q)+z\,(E_1^2(q)-\wp(q))/2+O(z^2)\,.
 }
 \end{array}
 \eq
 we also have
  \beq\label{a062}
  \begin{array}{l}
  \displaystyle{
 f(0,q)=-E_2(q)\,.
 }
 \end{array}
 \eq
 We use the following addition formulae:
\beq\label{a07}
  \begin{array}{c}
  \displaystyle{
  \phi(z_1, u_1) \phi(z_2, u_2) = \phi(z_1, u_1 + u_2) \phi(z_2 - z_1, u_2) + \phi(z_2, u_1 + u_2) \phi(z_1 - z_2, u_1)\,,
 }
 \end{array}
 \eq
and by applying $\p_{u_2}-\p_{u_1}$ we get
\beq\label{a071}
  \begin{array}{c}
  \displaystyle{
  \phi(z_1, u_1) f(z_2, u_2)-f(z_1,u_1)\phi(z_2,u_2) =
  }
  \\ \ \\
  \displaystyle{
  =
  \phi(z_1, u_1 + u_2) f(z_2 - z_1, u_2) - \phi(z_2, u_1 + u_2) f(z_1 - z_2, u_1)\,.
 }
 \end{array}
 \eq
Also,
\beq\label{a08}
  \begin{array}{c}
  \displaystyle{
 \phi(z,u_1)\phi(z,u_2)=\phi(z,u_1+u_2)\Big(E_1(z)+E_1(u_1)+E_1(u_2)-E_1(z+u_1+u_2)\Big)\,,
 }
 \end{array}
 \eq
\beq\label{a09}
  \begin{array}{c}
  \displaystyle{
  \phi(z, u_1) f(z,u_2)-\phi(z, u_2) f(z,u_1)=\phi(z,u_1+u_2)\Big(\wp(u_1)-\wp(u_2)\Big)=
 }
 \\ \ \\
   \displaystyle{
 =\phi(z,u_1+u_2)\Big(E_2(u_1)-E_2(u_2)\Big)=\phi(z,u_1+u_2)\Big(f(0,u_2)-f(0,u_1)\Big)\,.
 }
 \end{array}
 \eq
\beq\label{a10}
  \begin{array}{c}
  \displaystyle{
  \phi(z, u) \phi(z, -u) = \wp(z)-\wp(u)=E_2(z)-E_2(u)\,,
 }
 \end{array}
 \eq
\beq\label{a11}
  \begin{array}{c}
  \displaystyle{
  \phi(z, u) f(z, -u)-\phi(z, -u) f(z, u)=\wp'(u)\,.
 }
 \end{array}
 \eq
Define
 \beq\label{w209}
 \begin{array}{c}
  \displaystyle{
 \om_0=0\,,\quad
 \om_1=\frac{1}{2}\,,\quad
 \om_2=\frac{1+\tau}{2}\,,\quad
 \om_3=\frac{\tau}{2}
  }
 \end{array}
\eq
 and
 \beq\label{w210}
 \begin{array}{c}
  \displaystyle{
 \vf_{0}(z,\hbar)=\phi(z,\hbar)\,,
\qquad
 \vf_{1}(z,\hbar+\om_1)=\phi(z,\om_1+\hbar)\,,
 }
 \\ \ \\
  \displaystyle{
 \vf_{2}(z,\hbar+\om_2)=e^{\pi\imath z}\phi(z,\om_2+\hbar)\,,
\qquad
  \vf_{3}(z,\hbar+\om_3)=e^{\pi\imath z}\phi(z,\om_3+\hbar)\,.
 }
  \end{array}
 \eq
Then
 \beq\label{w211}
 \begin{array}{c}
  \displaystyle{
 \vf_{k}(z,\hbar+\om_k)=\frac{\theta_1'(0)\theta_{k+1}(z+\hbar)}{\theta_{1}(z)\theta_{k+1}(\hbar)}\,,
 \quad k=0,1,2,3
  }
 \end{array}
\eq
 The set of introduced functions is transformed under the action of matrix entering (\ref{q416}) as
  \beq\label{w215}
 \begin{array}{c}
  \left(\begin{array}{c}
 \vf_{0}(2z,u+\om_0)
 \\
  \vf_{1}(2z,u+\om_1)
 \\
  \vf_{2}(2z,u+\om_2)
 \\
  \vf_{3}(2z,u+\om_3)
 \end{array}\right)
 =
   \displaystyle{\frac12}
 \left(\begin{array}{cccc}
 1 & 1 & 1 & 1
 \\
  1 & 1 & -1 & -1
 \\
  1 & -1 & 1 & -1
 \\
  1 & -1 & -1 & 1
 \end{array}\right)
  \left(\begin{array}{c}
 \vf_{0}(2u,z+\om_0)
 \\
  \vf_{1}(2u,z+\om_1)
 \\
  \vf_{2}(2u,z+\om_2)
 \\
  \vf_{3}(2u,z+\om_3)
 \end{array}\right)\,.
 \end{array}
 \eq
  This transformation matrix can be written as
  \beq\label{q417}
  \begin{array}{c}
  \displaystyle{
I_{km}=\frac12\exp\Big(4\pi\imath\Big(\om_{m-1}\p_\tau\om_{k-1}
-\om_{k-1}\p_\tau\om_{m-1}\Big)\Big)\,,\quad k,m=1,...,4\,.
 }
 \end{array}
 \eq
It has the property
  \beq\label{q418}
  \begin{array}{c}
  \displaystyle{
I^{-1}=I\,.
 }
 \end{array}
 \eq
  We also use the widely known identity:
  \beq\label{w216}
 \begin{array}{c}
  \displaystyle{
\sum\limits_{a=0}^3\wp(z+\om_a)=4\wp(2z)\,.
  }
 \end{array}
 \eq
The quasi-periodic properties are as follows:
  \beq\label{a0621}
  \begin{array}{l}
  \displaystyle{
 \phi(z+1,u)=\phi(z,u)\,,\qquad  \phi(z+\tau,u)=e^{-2\pi\imath u}\phi(z,u)
 }
 \end{array}
 \eq
or
  \beq\label{a0622}
  \begin{array}{l}
  \displaystyle{
 \phi(z\pm 2\om_a,u)=e^{\mp4\pi\imath\p_\tau\om_a u}\phi(z,u)\,.
 }
 \end{array}
 \eq
For the theta function (\ref{q02}) we have
  \beq\label{a0623}
  \begin{array}{l}
  \displaystyle{
 \vth(z+2\om_a)=-e^{-4\pi\imath(z+\om_a)\p_\tau\om_a}\vth(z)
 }
 \end{array}
 \eq
or
  \beq\label{a0624}
  \begin{array}{l}
  \displaystyle{
 \vth(z+\om_a)=-e^{-4\pi\imath z\p_\tau\om_a}\vth(z-\om_a)\,.
 }
 \end{array}
 \eq
This provides
  \beq\label{a0625}
  \begin{array}{l}
  \displaystyle{
 E_2(z+2\om_a)=E_2(z)\,,\qquad E_1(z+2\om_a)=E_1(z)-4\pi\imath\p_\tau\om_a
 }
 \end{array}
 \eq
and
  \beq\label{a0626}
  \begin{array}{l}
  \displaystyle{
E_1(\om_a)=-2\pi\imath\p_\tau\om_a\,,\quad a\neq 0\,.
 }
 \end{array}
 \eq

\paragraph{Function $v(z,u)$ and related identities.}
Here we briefly discuss identities related to the function $v(z,u)$ (\ref{q413}).
It follows from (\ref{q419}) and (\ref{q415}) that
\beq\label{q419a}
v(z,u)v(-z,u)=\sum\limits_{a=0}^3\Big(\nu_a^2\wp(u+\om_a)-\bar\nu_a^2\wp(z+\om_a)\Big)=-v(z,u)v(z,-u)\,.
\eq
Let us prove the important identities (\ref{hident1}) and (\ref{hident3}).
The identity (\ref{hident1}) consists of four relations proportional to $\nu_0,...,\nu_3$.
Consider the one proportional to $\nu_0$:
  \beq\label{a931}
  \begin{array}{c}
  \displaystyle{
 \phi(2x,u)\phi(x+y,w-u)+\phi(2x,w)\phi(x-y,u-w)
 +\phi(2y,-u)\phi(x+y,u+w)=\,.
 }
 \\ \ \\
   \displaystyle{
=\phi(2y,w)\phi(x-y,u+w)\,.
 }
 \end{array}
 \eq
Using (\ref{a07}) we have
  \beq\label{a932}
  \begin{array}{c}
  \displaystyle{
\phi(x+y,w)\phi(x-y,u)=\phi(2x,u)\phi(x+y,w-u)-\phi(2x,w)\phi(y-x,w-u)\,.
 }
 \end{array}
 \eq
On the other hand for the same function we also have
  \beq\label{a933}
  \begin{array}{c}
  \displaystyle{
\phi(x+y,w)\phi(x-y,u)=\phi(2y,w)\phi(x-y,u+w)-\phi(2y,-u)\phi(x+y,u+w)\,.
 }
 \end{array}
 \eq
By comparing left hand sides of (\ref{a932}) and (\ref{a933}) one gets (\ref{a931}).
For the rest of the components of (\ref{hident1}) the proof is similar.

In order to prove (\ref{hident3}) let us compute the partial
 derivative $\partial_u+\partial_w$ in (\ref{hident1}):
\beq\label{hident2}
\begin{array}{c}
v'(x,u)\phi(x+y,w-u)+v'(x,w)\phi(x-y,u-w)-
\\ \\
-v'(y,-u)\phi(x+y,u+w)+2v(y,-u)f(x+y,u+w)=
\\ \\
=v'(y,w)\phi(x-y,u+w)+2v(y,w)f(x-y,u+w)
\end{array}
\eq
and then we put $x=0$ and $y:=-z$:
\beq\label{hident2a}
\begin{array}{c}
v'(0,u)\phi(-z,w-u)+v'(0,w)\phi(z,u-w)-v'(-z,-u)\phi(-z,u+w)+
\\ \\
+2v(-z,-u)f(-z,u+w)=v'(-z,w)\phi(z,u+w)+2v(-z,w)f(z,u+w)\,.
\end{array}
\eq
Finally, one should use an obvious antisymmetry property
\beq\label{antiphi}
\phi(z,u)=-\phi(-z,-u), \qquad v(-z,-u)=-v(z,u).
\eq
to get (\ref{hident3}).

%%%%%%%%%%%%%%%%%%%%%%%%%%%%%%%%%%%%%%%%%%%%%%%%%%%%%%%%%%%%%%%%%%%%%%%%%%%%%%%%%%%%%%%%%%%%%%%%%%%%%%%%%%
%%%%%%%%%%%%%%%%%%%%%%%%%%%%%%%%%%%%%%%%%%%%%%%%%%%%%%%%%%%%%%%%%%%%%%%%%%%%%%%%%%%%%%%%%%%%%%
%%%%%%%%%%%%%%%%%%%%%%%%%%%%%%%%%%%%%%%%%%%%%%%%%%%%%%%%%%%%%%%%%%%%%%%%%%%%%%%%%%%%%%%%%%%%%%

\subsection*{Acknowledgments}
We are grateful to O. Chalykh, An. Kirillov, A. Liashyk, V. Prokofev and I. Sechin for useful discussions.

This work was supported by the Russian Science Foundation under grant no. 25-11-00081,\\
https://rscf.ru/en/project/25-11-00081/ and performed at Steklov Mathematical Institute of Russian Academy of Sciences.

No conflict of interest exists for all participating authors. The manuscript has no associated data.

%This work was performed at the Steklov International Mathematical Center and supported by the Ministry of Science %and Higher Education of the Russian Federation (agreement no. 075-15-2022-265).
%\addcontentsline{toc}{section}{\hspace{6mm}Acknowledgments}

%This work was supported by the Russian Science Foundation under grant no. 19-11-00062,\\ %https://rscf.ru/en/project/19-11-00062/ .

%This work was performed at the Steklov International Mathematical Center and supported by the Ministry of Science %and Higher Education of the Russian Federation (agreement no. 075-15-2022-265).

%%%%%%%%%%%%%%%%%%%%%%%%%%%%%%%%%%%%%%%%%%%%%%%%%%%%%%%%%%%%%%%%%%%%%%%%%%%%%%%%%%%%%%%%%%%%%%
%%%%%%%%%%%%%%%%%%%%%%%%%%%%%%%%%%%%%%%%%%%%%%%%%%%%%%%%%%%%%%%%%%%%%%%%%%%%%%%%%%%%%%%%%%%%%%

\begin{small}

\end{small}


\begin{thebibliography}{99}
\addcontentsline{toc}{section}{References}

\bibitem{Ch2} Ph. Argyres, O. Chalykh, Y. L\"u,
{\it Inozemtsev System as Seiberg-Witten Integrable system},
 J. High Energ. Phys. 2021, 51 (2021);
	arXiv:2101.04505 [hep-th].



\bibitem{BB}
R.J. Baxter,
 {\it Partition function of the eight-vertex lattice model},
  Ann. Phys. 70 (1972) 193--228.

%A.A. Belavin,
%{\it Dynamical symmetry of integrable quantum systems},
%Nucl. Phys. B, 180 (1981) 189--200.

%\bibitem{BD}
%A. Belavin, V. Drinfeld,
%{\it Solutions of the classical Yang–Baxter equation for simple Lie algebras},
%Functional Analysis and Its Applications, 16:3 (1982) 159--180.

\bibitem{Ch1} O. Chalykh,
{\it Quantum Lax pairs via Dunkl and Cherednik operators},
Commun. Math. Phys. 369, 261–316 (2019);
	arXiv:1804.01766 [math.QA].

\bibitem{DP}
E. D' Hoker,  D.H. Phong,
 {\it Calogero-Moser Lax Pairs with Spectral Parameter for General Lie Algebras},
Nuclear Physycs B 530 (1998) 537--610; hep-th/9804124.


\bibitem{EFGR} A. Enciso, F. Finkel, A. González-Lopez, M.A. Rodríguez,
 {\it Haldane–Shastry spin chains of $BC_N$ type}, Nuclear physics B 707:3 (2005) 553--576;
 	arXiv:hep-th/0406054.

\bibitem{FP} G. Felder, V. Pasquier,
{\it A simple construction of elliptic R-matrices},
Lett. Math. Phys. 32 (1994) 167--171; arXiv:hep-th/9402011.

\bibitem{FK} S. Fomin, A.N. Kirillov,
{\it Quadratic algebras, Dunkl elements, and Schubert calculus},
Advances in geometry; Progress in Mathematics book series, Vol. 172
 (1999) 147--182.

M. Aguiar,
{\it Infinitesimal Hopf algebras},
New trends in Hopf algebra theory,
Contemporary Mathematics 267
(2000) 1–-29.

\bibitem{GW} F. Gesztesy, R. Weikard,
{\it Treibich-Verdier potentials and the stationary (m) KdV hierarchy},
Mathematische Zeitschrift, 219 (1995) 451--476.

V. Prokofev, unpublished.

\bibitem{GSZ}
A. Grekov, I. Sechin, A. Zotov,
 {\it Generalized model of interacting integrable tops},  JHEP 10 (2019) 081;
 	arXiv:1905.07820 [math-ph].

\bibitem{GrZ18} A. Grekov, A. Zotov,
{\it On R-matrix valued Lax pairs for Calogero–Moser models},
J. Phys. A: Math. Theor., 51 (2018), 315202;
	arXiv:1801.00245 [math-ph].

\bibitem{Hikami} K. Hikami, {\it Boundary K-matrix, elliptic Dunkl operator and quantum
many-body systems}, J. Phys. A: Math. Gen. 29 (1996) 2135--2147.




\bibitem{IK} T. Inami, H. Konno,
{\it Integrable XYZ spin chain with boundaries},
J. Phys. A Math. Gen. 27 (1994) L913--L918.


\bibitem{Inoz89} V.I. Inozemtsev,
 {\it Lax representation with spectral parameter on a torus for integrable particle systems},
Lett. Math. Phys. 17 (1989) 11--17.

\bibitem{Inoz90} V.I. Inozemtsev,
 {\it On the connection between the one-dimensional S=1/2 Heisenberg chain and Haldane-Shastry model},
Journal of Statistical Physics, 59 (1990) 1143--1155.

J. Dittrich, V. Inozemtsev,
 {\it The commutativity of integrals of motion for quantum spin chains and elliptic functions identities},
Regular and Chaotic Dynamics, 13:1 (2008) 19--26; 	arXiv:0711.1973 [math-ph].

O. Chalykh, {\it Integrability of the Inozemtsev spin chain}, arXiv:2407.03276 [nlin.SI].


\bibitem{IKT} A.P. Isaev, A.N. Kirillov, V.O. Tarasov,
{\it Bethe subalgebras in affine Birman--Murakami--Wenzl algebras and flat connections for q-KZ equations},
 J. Phys. A: Math. Theor. 49 (2016) 204002; arXiv:1510.05374 [math.RT].

\bibitem{KH} Y. Komori, K. Hikami,
{\it Quantum integrability of the generalized elliptic Ruijsenaars models},
1997 J. Phys. A: Math. Gen. 30, 4341.

Y. Komori, K. Hikami,
{\it Conserved operators of the generalized elliptic Ruijsenaars models},
J. Math. Phys. 39, 6175--6190 (1998).


\bibitem{KH2} Y. Komori, K. Hikami, {\it Elliptic K-matrix associated with Belavin's symmetric R-matrix}, Nuclear Physics B 494.3 (1997): 687-701.

\bibitem{Kir} Anatol N. Kirillov, {\it On Some Quadratic Algebras I $\frac12$:
Combinatorics of Dunkl and Gaudin Elements, Schubert, Grothendieck, Fuss-Catalan, Universal Tutte and Reduced Polynomials},
SIGMA 12 (2016), 002; arXiv:1502.00426 [math.RT].

\bibitem{Krich1}
I.  Krichever,
{\it Elliptic solutions of the Kadomtsev--Petviashvili equation and integrable systems of particles},
Funct.  Anal.  Appl., 14:4 (1980) 282--290.

\bibitem{Lam} J. Lamers, V. Pasquier, D. Serban,
 {\it Spin-Ruijsenaars, q-deformed Haldane-Shastry and Macdonald polynomials},
 Commun. Math. Phys. 393, pages 61-150 (2022);
	arXiv:2004.13210 [math-ph].

R. Klabbers, J. Lamers,
{\it Landscapes of integrable long-range spin chains},
	arXiv:2405.09718 [math-ph].

\bibitem{LOZ14}
A. Levin, M. Olshanetsky, A. Zotov,
{\it Planck constant as spectral parameter in integrable systems and KZB equations},
JHEP 10 (2014) 109; arXiv:1408.6246 [hep-th].

\bibitem{LOZ15}
A.M. Levin, M.A. Olshanetsky, A.V. Zotov,
{\it Quantum Baxter–Belavin R-matrices and multidimensional Lax pairs for Painlev\'e VI},
    Theoret. and Math. Phys., 184:1 (2015), 924--939;	arXiv:1501.07351 [math-ph].

\bibitem{LOZ16} 	A. Levin, M. Olshanetsky, A. Zotov,
 {\it Yang–Baxter equations with two Planck constants},
  J. Phys. A: Math. Theor., 49:1 (2016), 14003 , 19 pp.,
  Exactly Solved Models and Beyond: a special issue
  in honour of R. J. Baxter's 75th birthday, arXiv: 1507.02617 [math-ph].

\bibitem{LRS} A.~Liashyk, N.~Reshetikhin, I.~Sechin: {\it Quantum integrable systems on a classical integrable background.}  {arXiv:2405.17865 [math-ph]} (2024).


\bibitem{MZ} M. Matushko, A. Zotov,
{\it Anisotropic spin generalization of elliptic Macdonald–Ruijsenaars operators and R-matrix identities},
Ann. Henri Poincar\'e, 24 (2023), 3373–3419;
	arXiv:2201.05944 [math.QA].

M. Matushko, A. Zotov,
{\it Elliptic generalisation of integrable q-deformed anisotropic Haldane–Shastry long-range spin chain},
Nonlinearity, 36:1 (2023), 319;
	arXiv:2202.01177 [math-ph].

\bibitem{OP} M.A. Olshanetsky, A.M. Perelomov,
 {\it Completely integrable Hamiltonian systems connected with semisimple Lie  algebras},
Inventiones mathematicae, 37:2 (1976) 93–-108.

M.A. Olshanetsky, A.M. Perelomov,
 {\it Classical integrable finite-dimensional systems related to Lie algebras},
Physics Reports, 71 (1981) 313--400.

\bibitem{ORS} A. Odesskii, V. Rubtsov, V. Sokolov,
{\it Parameter-dependent associative Yang-Baxter equations and Poisson brackets},
International Journal of Geometric Methods in Modern Physics, Vol. 11, No. 09, 1460036 (2014);
arXiv:1311.4321 [math-ph].

\bibitem{Pol} A. Polishchuk,
 {\it Classical Yang–Baxter equation and the $A^\infty$-constraint},
Advances in Mathematics 168:1 (2002)  56–-95;
arXiv:math/0008156 [math.AG].

A. Polishchuk,
 {\it $A_\infty$-structures on an elliptic curve},
 Commun. Math. Phys. 247 (2004) 527–-551;
arXiv: math/0001048 [math.AG].

\bibitem{P} A. P. Polychronakos: {\it Lattice integrable systems of Haldane--Shastry type.} Phys. Rev. Lett. \textbf{70}(15), 2329--2331 (1993),  arXiv:hep-th/9210109.

\bibitem{SeZ18} I. Sechin, A. Zotov,
{\it R-matrix-valued Lax pairs and long-range spin chains},
Phys. Lett. B, 781 (2018), 1–7; arXiv:1801.08908 [math-ph].

\bibitem{SU}
Y. Shibukawa, K. Ueno. {\it Completely $\mathbb Z$ symmetric R matrix}, Letters in Mathematical Physics 25 (1992): 239-248.

\bibitem{Skl-refl} E. Sklyanin,
{\it Boundary conditions for integrable quantum systems},
J. Phys. A: Math. Gen. 21 (1988) 2375--2389.

\bibitem{Ta} K. Takasaki,
{\it Elliptic Calogero-Moser Systems and Isomonodromic Deformations},
	J. Math. Phys. 40 (1999) 5787--5821;
	arXiv: math/9905101 [math.QA].

\bibitem{Weil} A. Weil, {\it Elliptic functions according to Eisenstein and
Kronecker}, Springer-Verlag, %Berlin- Heidelberg-New York,
 (1976).

 D. Mumford, {\it Tata Lectures on Theta I, II},
Birkh\"auser, Boston, Mass. (1983, 1984).

%\bibitem{ZZ} A. Zabrodin, A. Zotov,
%{\it Field analogue of the Ruijsenaars-Schneider model},
%JHEP 07 (2022) 023;	arXiv: 2107.01697 [math-ph].

\bibitem{Z04} A. Zotov,
{\it Elliptic linear problem for the Calogero-Inozemtsev model and Painlev\'e VI equation},
 Lett. Math. Phys., 67:2 (2004), 153--165; 	arXiv:hep-th/0310260.

\bibitem{Z18}  A.V. Zotov, {\it Calogero–Moser model and R-matrix identities},
 Theoret. and Math. Phys., 197:3 (2018), 1755--1770.

A.M. Levin, M.A. Olshanetsky, A.V. Zotov,
 {\it Classification of isomonodromy problems on elliptic curves},
 Russian Math. Surveys, 69:1 (2014) 35--118; 	arXiv:1311.4498 [math-ph].

 A.M. Levin, M.A. Olshanetsky, A.V. Zotov,
 {\it Painlev\'e VI, rigid tops and reflection equation}, Commun. Math. Phys., 268:1 (2006) 67--103; arXiv:math/0508058 [math.QA].

%\bibitem{ZFourier} A. Zotov, {\it Relativistic elliptic matrix tops and finite Fourier transformations},
	%Modern Physics Letters A, 32 (2017) 1750169, 	arXiv:1706.05601 [math-ph].

\end{thebibliography}
\end{document}